\documentclass{sig-alternate-10pt}

\usepackage{graphicx}
\usepackage{amsmath,amssymb,makeidx,mathrsfs}
\usepackage{algpseudocode}
\usepackage[font=small,format=plain,labelfont=bf,up,textfont=bf,up]{caption}
\usepackage[caption=false,font=small]{subfig}
\usepackage{xspace}
\usepackage{cite}
\usepackage{thumbpdf}
\usepackage{tikz}
\usepackage[nomessages]{fp}
\usepackage{ifthen}
\usepackage{listings}

\usepackage{xcolor}
\definecolor{pennblue}{cmyk}{1,.65,0,.3}
\definecolor{pennred}{cmyk}{0,1,.65,.34}
\usepackage[colorlinks=true,linkcolor={pennblue},citecolor={pennred},urlcolor={pennblue}]{hyperref}
\hypersetup{breaklinks=true,pdfstartview=FitH,pdfpagelayout=SinglePage}

\usetikzlibrary{arrows,arrows.meta,automata,decorations.pathreplacing,decorations.markings,shapes.geometric,fit}

\newcommand{\smartparagraph}[1]{\noindent{\bf #1}\ }
\newcommand{\eg}{{\it e.g.\xspace}}
\newcommand{\ie}{{\it i.e.\xspace}}

\newcommand{\eat}[1]{}
\newcommand{\schedulingheuristic}{{\sc ez-Schedule}\xspace}

\definecolor{shadecolor}{rgb}{0.9,0.9,0.9}
\definecolor{heraldBlue}{rgb}{0.0,0.0,0.8}
\definecolor{heraldRed}{rgb}{0.8,0.0,0.0}
\definecolor{heraldGray}{rgb}{0.4,0.4,0.4}
\definecolor{heraldBlack}{rgb}{0.0,0.0,0.0} 
\definecolor{heraldGreen}{rgb}{0.0,0.4,0.0} 

\def\NOTES{0}
\def\TR{1}

\if \NOTES 1
\newcommand{\mcnote}[1]{\textcolor{heraldBlue}{\small \bf [MC: #1]}}
\newcommand{\tnnote}[1]{\textcolor{heraldRed}{\small \bf [TN: #1]}}
\newcommand{\ed}[1]{\textsf{\textcolor{red}{[mch: #1]}}}

\newcommand{\spacebeforeparagraph}{\vspace{.1in}}
\else
\newcommand{\mcnote}[1]{}
\newcommand{\tnnote}[1]{}

\newcommand{\ed}[1]{}
\newcommand{\spacebeforeparagraph}{}
\fi

\tikzset{
	every label/.style={draw=none, inner sep=0mm},
	neutral/.style={draw=black},
	inf/.style={draw=black, fill=black!30!red},
	vertex/.style={circle, draw=black!70, fill=white, thick,inner sep=0pt,minimum size=4pt},
	controller/.style={circle, draw=blue!50, fill=yellow!20, thick,inner sep=2.5pt,minimum size=4pt},
	resource/.style={rectangle, draw=black, fill=white, thick,inner sep=4pt,minimum size=4pt},
	move/.style={circle, draw=black!70, fill=yellow!20, thick,inner sep=4pt,minimum size=4pt},
	sus/.style={draw=black, fill=black!30!yellow},
	killed/.style={draw=black, fill=gray!70},
	det/.style={draw=black, fill=green!30},
	cutAction/.style={draw=black, fill=white},
	act/.style={-},
	cut/.style={-, thick, color=brown},
	ctrl/.style={-,dashed},
	att/.style={->, very thick, color=black!30!red},
	cross/.style={cross out, color=brown, draw,-,very thick},
	deact/.style={->, ultra thick, color=gray},
	set arrow inside/.code={\pgfqkeys{/tikz/arrow inside}{#1}},
	set arrow inside={end/.initial=>, opt/.initial=},
	/pgf/decoration/Mark/.style={
		mark/.expanded=at position #1 with
		{
			\noexpand\arrow[\pgfkeysvalueof{/tikz/arrow inside/opt}]{\pgfkeysvalueof{/tikz/arrow inside/end}}
		}
	},
	arrow inside/.style 2 args={
		set arrow inside={#1},
		postaction={
			decorate,decoration={
				markings,Mark/.list={#2}
			}
		}
	},
	moveblank/.style={draw=white,
		regular polygon,
		regular polygon sides=3,
		fill=white,
		thick,
		node distance=1pt,
		minimum height=1pt,
		inner sep=1pt,minimum size=2pt}
}

\newcommand{\ezsegway}{ez-Segway\xspace}
\newcommand{\dionysus}{Dionysus\xspace}
\newcommand{\MsgInstallUpdate}{\texttt{InstallUpdate}\xspace}
\newcommand{\MsgRemoving}{\texttt{Removing}\xspace}
\newcommand{\MsgGoodToMove}{\texttt{GoodToMove}\xspace}
\newcommand{\MsgCoherent}{\texttt{Coherent}\xspace}

\newcommand{\NetworkConfig}{\ensuremath{\mathbb{C}}}

\newcommand{\NwSwitch}{\ensuremath{s}}
\newcommand{\Link}[2]{\ensuremath{\ell_{#1,#2}}}
\newcommand{\LinkSet}{\ensuremath{\mathbb{L}}}
\newcommand{\LinkSetDependent}{\ensuremath{L}}

\newcommand{\Packet}{\ensuremath{\mathit{pk}}}
\newcommand{\Tag}{\ensuremath{tag}}

\newcommand{\nextSW}[3]{\ensuremath{next_{#1,#2}(#3)}}

\newcommand{\Flow}{\ensuremath{\mathit{F}}}

\newcommand{\Traffic}{\ensuremath{v}}

\newcommand{\MovePathSet}{\ensuremath{\Pi}}
\newcommand{\MovePath}{\ensuremath{\pi}}

\newcommand{\Before}{\ensuremath{old}}

\newcommand{\NwSwitchSet}{\ensuremath{\mathbb{S}}}
\newcommand{\NwPairSwitchSet}{\ensuremath{\mathbb{P}}}

\newcommand{\DependenceGraph}[4]{\ensuremath{\mathbb{G}(#1,#2,#3,#4)}}

\newcommand{\Null}{\bot}

\newcommand{\SwInit}{\ensuremath{init}}

\newcommand{\basicmigration}{\textsc{Basic-Update}\xspace}
\newcommand{\segmentedmigration}{\textsc{Segmented-Update}\xspace}
\newcommand{\bfs}{BFS\xspace}

\newtheorem{theorem}{Theorem}
\newtheorem{lemma}{Lemma}

\if \TR 0
\newfont{\mycrnotice}{ptmr8t at 7pt}
\newfont{\myconfname}{ptmri8t at 7pt}
\let\crnotice\mycrnotice%
\let\confname\myconfname%
\permission{Permission to make digital or hard copies of part or all of this work for personal or classroom use is granted without fee provided that copies are not made or distributed for profit or commercial advantage and that copies bear this notice and the full citation on the first page. Copyrights for third-party components of this work must be honored. For all other uses, contact the Owner/Author(s). Copyright is held by the owner/author(s).}
\CopyrightYear{2017}
\setcopyright{acmlicensed}
\conferenceinfo{SOSR '17,}{April 3--4, 2017, Santa Clara, CA, USA}
\isbn{978-1-4503-4947-5/17/04}
\acmPrice{\$15.00}
\doi{http://dx.doi.org/10.1145/3050220.3050224}
\fi

\begin{document}

\newfont{\thisttlfnt}{phvb8t at 18pt}

\if \TR 1
  \title{Decentralized Consistent Network Updates in SDN\\ with ez-Segway}
\else
  \title{Decentralized Consistent Updates in SDN}
\fi
\numberofauthors{3}
\author{
        \alignauthor Thanh Dang Nguyen\thanks{Work performed at Universit\'e catholique de Louvain.}\\
        \affaddr{University of Chicago}
        \alignauthor Marco Chiesa\\
        \affaddr{\mbox{Universit\'e catholique de Louvain}}
        \alignauthor Marco Canini\footnotemark[1]\\
	\affaddr{KAUST}
	}
\date{}

\maketitle
\pagestyle{empty}
\begin{abstract}

We present {\em ez-Segway}, a decentralized mechanism to consistently and
quickly update the network state while preventing forwarding anomalies (loops
and black-holes) and avoiding link congestion. In our design, the centralized
SDN controller only pre-computes information needed by the switches during the
update execution. This information is distributed to the switches, which use
partial knowledge and direct message passing to efficiently realize the update.
This separation of concerns has the key benefit of improving update performance
as the communication and computation bottlenecks at the controller are removed.
Our evaluations via network emulations and large-scale simulations 
demonstrate the efficiency of \ezsegway, which compared to a centralized 
approach, improves network update times by up to $45\%$ and $57\%$ at 
the median and the  $99^{th}$ percentile, respectively. 
A deployment of a system 
prototype 
in a real
OpenFlow switch and an implementation in P4 demonstrate the feasibility and low
overhead of implementing simple network update functionality within switches.

\end{abstract}\if \TR 0


\fi

\section{Introduction}

Updating data plane state to adapt to dynamic conditions is a fundamental
operation in all centrally-controlled networks. Performing updates as quickly as
possible while preserving certain consistency properties (like loop, black-hole
or congestion freedom) is a common challenge faced in many recent SDN systems
(e.g.,~\cite{b4, hong13, Levin.ATC14, Soule.CONEXT14}).

Network updates are inherently challenging because rule-update operations happen
across unsynchronized devices and the consistency properties impose dependencies
among operations that must be respected to avoid forwarding anomalies and
worsened network performance~\cite{Jin+DSN+2014,survey-network-updates}. Yet,
performing updates as fast as possible is paramount in a variety of scenarios
ranging from performance to fault tolerance to
security~\cite{Mahajan.HotNets.2013, Jin+DSN+2014, Peresini.HOTSDN14} and is a
crucial requirement for ``five-nines'' availability of carrier-grade and mobile
backhaul networks~\cite{time4-infocom-16}.

Ideally, a network would have the capability to instantaneously update its
network-wide state while preserving consistency. Since updates cannot be applied
at the exact same instant at all switches, recent work~\cite{time4-infocom-16}
explored the possibility of leveraging clock-synchronization techniques to
minimize the transition time. This requires just a \emph{single round of
communication} between the controller and switches. However, with this
``all-in-one-shot'' update style, packets in flight at the time of the change
can violate consistency properties due to the complete lack of
coordination~\cite{Jin+DSN+2014} or are dropped as a result of scheduling errors
due to clock synchronization imprecisions~\cite{time4-infocom-16}.

The bulk of previous work in this area~\cite{Reitblatt+ANU+2012,
Peresini.HOTSDN14, katta2013incremental, Jin+DSN+2014, Liu+ZUD+2013,
McClurg.PLDI15, Ludwig+PODC2015} focused on \emph{maintaining consistency
properties}. However, all these approaches require multiple rounds of
communications between the controller (which drives the update) and the switches
(which behave as remote passive nodes that the controller writes state to and is
notified from). This controller-driven update process has four important
drawbacks:

First, because the controller is involved with the
installation of every rule, {\em the update time is inflated by inherent delays}
affecting communication between controller and switches. As a result, even with
state-of-the-art approaches~\cite{Jin+DSN+2014}, a network update typically
takes seconds to be completed (results show $99^{th}$ percentiles as high as 4
seconds).

Second, because scheduling updates is computationally
expensive~\cite{Jin+DSN+2014, Peresini.HOTSDN14, McClurg.PLDI15}, {\em the
update time is slowed down by a centralized computation}.

Third, because the controller can only react to network dynamics (\eg,
congestion) at control-plane timescales, {\em the update process cannot 
quickly adapt to current data plane conditions}.

Fourth, in real deployments where the controller is distributed and
switches are typically sharded among controllers, network \emph{updates require
additional coordination overhead among different controllers}. For wide-area
networks, this additional synchronization can add substantial
latency~\cite{ICONA}.

In this paper, we present ez-Segway, a new mechanism for network updates. Our
key insight is to involve the switches as active participants in achieving fast
and consistent updates. The controller is responsible for computing the intended
network configuration, and for identifying ``flow segments'' (parts of an update
that can be updated independently and in parallel) and dependencies among
segments. The update is realized in a \emph{decentralized} fashion by the
switches, which execute a network update by scheduling update operations based
on the information received by the controller and messages passed among
neighboring switches. This allows every switch to update its local forwarding
rules as soon as the update dependencies are met (\ie, when a rule can only be
installed after dependent rules are installed at other switches), without any
need to coordinate with the controller.

Being decentralized, our approach differs significantly from prior work. It
achieves the main benefit of clock-synchronization mechanisms, which incur only
a single round of controller-to-switches communication, with the guarantees
offered by coordination to avoid forwarding anomalies and congestion. To the
best of our knowledge, we are the first to explore the benefits of delegating
network updates' coordination and scheduling functionalities to the switches.
Perhaps surprisingly, we find that the required functionality can be readily
implemented in existing programmable switches (OpenFlow and P4~\cite{p4}) with
low overhead.

Preventing deadlocks is an important challenge of our update mechanism. To
address this problem, we develop two techniques: (i) \emph{flow segmentation},
which allows us to update different flow segments independently of one another,
thus reducing dependencies, and (ii) \emph{splitting volume}, which divides a
flow's traffic onto its old and new paths, thus preventing link congestion.

We show that our approach leads to faster network updates, reduces the number of
exchanged messages in the network, and has low complexity for the scheduling
computation. Compared to state-of-the-art centralized approaches
(\eg, Dionysus~\cite{Jin+DSN+2014}) configured with the most favorable controller
location in the network,  ez-Segway improves network update time by up to 
$45\%$ and $57\%$ at the median and the  $99^{th}$ percentile, respectively,
and it reduces message overhead by $65\%$. To put these results into
perspective, consider that Dionysus improved the network update time by up to a factor
of $2$ over the previous state-of-the-art technique, SWAN~\cite{hong13}.
Overall, our results show significant update time reductions, with speed of light
becoming the main limiting factor.

Our contributions are as follows:\\
	$\bullet$ We present \ezsegway (\S\ref{sec:algo}),
	a consistent update scheduling mechanism that runs on switches,
	initially coordinated by a centralized SDN controller.\\
	$\bullet$ We assess our prototype implementation
	(\S\ref{sec:implementation}) by running a comprehensive set of emulations
	(\S\ref{sec:evaluation}) and simulations (\S\ref{sec:simulation}) on various topologies and traffic patterns.\\
	$\bullet$ We validate feasibility (\S\ref{sec:micro-benchmark}) by running
	our system prototype on a real SDN switch and present microbenchmarks that
	demonstrate low computational overheads. Our prototype
	is available as open source at\\ {\footnotesize \url{https://github.com/thanh-nguyen-dang/ez-segway}}.

\section{Network Update Problem}
\label{sec:problem}

We start by formalizing the network update problem and the properties we are
interested in. The network consists of switches $\NwSwitchSet{=}\{\NwSwitch_i\}$
and directed links $\LinkSet{=}\{\Link{i}{j}\}$, in which $\Link{i}{j}$ connects
$s_i$ to $s_j$ with a certain capacity.

\smartparagraph{Flow modeling.}
We use a standard model for characterizing flow traffic volumes as
in~\cite{hong13, Jin+DSN+2014,survey-network-updates}. A flow $\Flow$ is 
an aggregate of packets
between an ingress switch and an egress switch. Every flow is associated with a
{\em traffic volume} $\Traffic_{\Flow}$. In practice, this volume could be an
estimate that the controller computes by periodically gathering switch
statistics~\cite{Jin+DSN+2014,survey-network-updates} or based on an 
allocation of bandwidth
resources~\cite{b4,survey-network-updates}. The {\em forwarding state} of a 
flow consists of an exact
match rule that matches all packets of the flow. As in previous
work~\cite{Jin+DSN+2014}, we assume that flows can be split among paths by means
of weighted load balancing schemes such as WCMP or OpenFlow-based approaches
like Niagara~\cite{Kang.CoNEXT15}.

\smartparagraph{Network configurations and updates.}
A \emph{network configuration} $\NetworkConfig$ is the collection of all
forwarding states that determine what packets are forwarded between any pair of
switches and how they are forwarded (\eg, match-action flow table rules in
OpenFlow). Given two network configurations $\NetworkConfig, \NetworkConfig'$, a
{\em network update} is a process that replaces the current network
configuration $\NetworkConfig$ by the target one $\NetworkConfig'$.

\smartparagraph{Properties.}
We focus on these three properties of network updates: $(i)$ {\em black-hole
freedom}: No packet is unintendedly dropped in the network; $(ii)$ {\em
loop-freedom}: No packet should loop in the network; $(iii)$ {\em
congestion-freedom}: No link should be loaded with a traffic greater than its
capacity. These properties are the same as the ones studied
in~\cite{Mahajan.HotNets.2013,survey-network-updates}.

\smartparagraph{Update operations.}
Due to link capacity limits and the inherent difficulty in synchronizing the
changes at different switches, the link load during an update could get
significantly higher than that before or after the update and all flows of a
configuration cannot be moved at the same time. Thus, to minimize disruptions,
it is necessary to decompose a network update into a set of {\em update
operations} $\MovePathSet$.

Intuitively, an update operation $\MovePath$ denotes the operation necessary to
move a flow $\Flow$ from the old to the new configuration: in the context of a
single switch, this refers to the addition or deletion of $\Flow$'s forwarding
state for that switch.

Thus, unlike per-packet consistent updates~\cite{Reitblatt+ANU+2012}, we allow a
flow to be routed through a mix of the old and new configuration, unless other
constraints make it impossible (\eg, a service chain of middle-boxes must be
visited in reversed order in the two configurations).

\smartparagraph{Dependency graph.}
To prevent a violation of our properties, update operations are constrained in
the order of their execution. These dependencies can be described with the help
of a {\em dependency graph}, which will be  explained in detail in
\S\ref{sec:algo}. At any time, only the update operations whose dependencies are
met in the dependency graph are  possibly executed. That leads the network to
transient intermediate configurations. The {\em update is successful} when the
network is transformed from the current configuration $\NetworkConfig$ to the
target configuration $\NetworkConfig'$ such that for all intermediate
configurations, the aforementioned three properties are preserved.
\section{Overview}
\label{sec:overview}

Our goal is to improve the performance of updating network state while avoiding
forwarding loops, black-holes, and congestion. In contrast with prior approaches
-- all of which use a controller to plan and coordinate the updates -- we
explore the question: can the network update problem be solved in a
decentralized fashion wherein switches are delegated the task of implementing
network updates?
 
\subsection{Decentralizing for fast updates}

Consider the example of seven switches $s_1, \cdots, s_7$ shown in
Figure~\ref{fig.big-network}. Assume each link has $10$ units of capacity and
there are three flows $\Flow_R,\Flow_G,\Flow_B$, each of size 5. This means that
every link can carry at most $2$ flows at the same time. We denote a path
through a sequence of nodes $s1,\dots,s_n$ by $(s1\dots s_n)$.

\begin{figure}[ptb]
	\centering
	\scalebox{0.8}{
		\begin{tikzpicture}[scale=1.6]
			\label{fig.notation-flow}
			\node[above] at (0,0.8) {$F_R{:}5$};
			\draw[red!75!white,thick] plot [smooth,tension=0.5] coordinates { (0.25, 0.95) (0.6, 0.95) } [arrow inside={end=stealth,opt={red!75!white,scale=2}}{0.99}];
			
			\node[above] at (0,0.5) {$F_G{:}5$};
			\draw[green!50!black,thick] plot [smooth,tension=0.5] coordinates { (0.25, 0.65) (0.6, 0.65) } [arrow inside={end=stealth,opt={green!50!black,scale=2}}{0.99}];
			
			\node[above] at (0,0.2) {$F_B{:}5$};
			\draw[blue,thick] plot [smooth,tension=0.5] coordinates { (0.25, 0.35) (0.6, 0.35) } [arrow inside={end=stealth,opt={blue,scale=2}}{0.99}];
			\end{tikzpicture}
		}~
		\subfloat[\scriptsize{Configuration $\NetworkConfig_1$}]{
			\label{fig.c1}
			\scalebox{0.95}{
				\begin{tikzpicture}[scale=0.75]
					\node [vertex] (1) at (-0.5,0) {$s_1$};
					\node [vertex] (2) at (0.5,0) {$s_2$};
					\node [vertex] (3) at (1.5,0) {$s_3$};
					\node [vertex] (4) at (2.5,0) {$s_4$};
					\node [vertex] (5) at (3.5,0) {$s_5$};
					\node [vertex] (6) at (1,1) {$s_6$};
					\node [vertex] (7) at (2,1) {$s_7$};
					\draw[act] (1) to (2);
					\draw[act] (2) to (3);
					\draw[act] (2) to (6);
					\draw[act] (3) to (4);
					\draw[act] (3) to (6);
					\draw[act] (3) to (7);
					\draw[act] (4) to (5);
					\draw[act] (4) to (7);
					
					\draw[blue, thick] plot [smooth,tension=0.5] coordinates { (-0.4,0.2) (0.5,0.2) (0.8,1.1) (1.15,1.15) (1.5,0.2) (2.5,0.2) } [arrow inside={end=stealth,opt={blue,scale=2}}{0.99}];
					
					\draw[green!40!black, thick] plot [smooth,tension=0.6] coordinates { (0.68,0.1) (1,0.7) (1.32,0.1) } [arrow inside={end=stealth,opt={green!40!black,scale=2}}{0.99}];
		
					\draw[red!75!white, thick] plot [smooth,tension=0.5]  coordinates { (0.65,0.15) (1.5,0.2) (1.8,1.1) (2.15,1.15) (2.5,0.2) (3.5, 0.2)} [arrow inside={end=stealth,opt={red!50!white,scale=2}}{0.99}];
				\end{tikzpicture}
			}
		}
		\subfloat[\scriptsize{Configuration $\NetworkConfig_1'$}]{
			\label{fig.c2}
			\scalebox{0.95}{
				\begin{tikzpicture}[scale=0.75]
					\node [vertex] (1) at (-0.5,0) {$s_1$};
					\node [vertex] (2) at (0.5,0) {$s_2$};
					\node [vertex] (3) at (1.5,0) {$s_3$};
					\node [vertex] (4) at (2.5,0) {$s_4$};
					\node [vertex] (5) at (3.5,0) {$s_5$};
					\node [vertex] (6) at (1,1) {$s_6$};
					\node [vertex] (7) at (2,1) {$s_7$};
					\draw[act] (1) to (2);
					\draw[act] (2) to (3);
					\draw[act] (2) to (6);
					\draw[act] (3) to (4);
					\draw[act] (3) to (6);
					\draw[act] (3) to (7);
					\draw[act] (4) to (5);
					\draw[act] (4) to (7);
					
					\draw[red!75!white, thick] plot [smooth,tension=0.5]  coordinates { (0.5,0.2) (0.8,1.1) (1.15,1.15) (1.5,0.25) (2.5, 0.2) (3.5, 0.2)} [arrow inside={end=stealth,opt={red!50!white,scale=2}}{0.99}];
					
					\draw[green!40!black, thick] plot [smooth,tension=0.6] coordinates { (0.68,0.1) (1,0.7) (1.32,0.1) } [arrow inside={end=stealth,opt={green!40!black,scale=2}}{0.99}];

					\draw[blue, thick] plot [smooth,tension=0.5] coordinates { (-0.4,0.2) (1.45,0.25) (1.8,1.1) (2.15,1.15) (2.5,0.2) } [arrow inside={end=stealth,opt={blue,scale=2}}{0.99}];

				\end{tikzpicture}
			}
		}
	\caption{An example network update. Three flows $\Flow_R,\Flow_G,\Flow_B$ are to be updated in the new configuration $\NetworkConfig_1'$. Update operations must be ordered carefully in order to preserve consistency and avoid congestion.}
	\label{fig.big-network}
\end{figure}
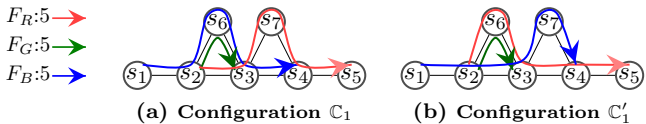
The network configuration needs to be updated from $\NetworkConfig_1$ to
$\NetworkConfig_1'$. Note that we cannot simply transition to
$\NetworkConfig_1'$ by updating all the switches at the same time. Since
switches apply updates at different times, such a strategy can neither ensure
congestion freedom nor loop- and black-hole- freedom. For example, if $s_2$ is
updated to forward $\Flow_R$ on link $\Link{2}{6}$ before the forwarding state
for $\Flow_R$ is installed at $s_6$, this results in a temporary black-hole.
Moreover, if $s_2$ forwards $\Flow_R$ on link $\Link{2}{6}$ before $\Flow_B$ is
moved to its new path, then $\Link{2}{6}$ becomes congested.

Ensuring that the network stays congestion free and that the consistency of
forwarding state is not violated requires us to carefully plan the order of
update operations across the switches. In particular, we observe that certain
operations depend on other operations, leading potentially to long chains of
dependencies for non trivial updates.
This fact implies that a network update
can be slowed down by two critical factors: $(i)$ the amount of computation
necessary to find an appropriate update schedule, and $(ii)$ the inherent
latencies affecting communication between controller and switches summed over
multiple rounds of communications to respect operation dependencies.

But is it necessary to have the controller involved at every step of a network
update? We find that it is possible to achieve fast and consistent updates by
minimizing controller involvement and delegating to the switch the tasks of
scheduling and coordinating the update process.

We leverage two main insights to achieve fast, consistent updates. Our first
insight is that we can complete an update faster by using \emph{in-band
messaging} between switches instead of coordinating the update at the
controller, which pays the costs of higher latency. Our second insight is that
it is not always necessary to move a flow as a whole. We can complete the update
faster by using \emph{segmentation}, wherein different ``segments'' of a flow
can be updated in parallel. For example, in Figure~\ref{fig.big-network}, both
flows $\Flow_R$ and $\Flow_B$ can be decomposed in two independent parts: the
first between $s_2$ and $s_3$, and the second between $s_3$ and $s_4$.

\smartparagraph{Distributed update execution.}
Before we introduce our approach in detail, we illustrate the execution of a
decentralized network update for the example of Figure~\ref{fig.big-network}.
Initially, the controller sends to every switch a message containing the current
configuration $\NetworkConfig_1$, and target configuration
$\NetworkConfig_1'$.\footnote{We revisit later what the controller sends
precisely.} This information allows every switch to compute \emph{what}
forwarding state to update (by knowing which flows traverse it and their sizes)
as well as \emph{when} each state update should occur (by obeying operation
dependencies while coordinating with other switches via in-band messaging).

In the example, switch $s_2$ infers that link $\Link{2}{3}$ has enough capacity
to carry $\Flow_B$ and that its next hop switch, $s_3$, is already capable of
forwarding $\Flow_B$ (because the flow traverses it in both the old and new
configuration). Hence, $s_2$ updates its forwarding state so as to move
$\Flow_B$ from path $(s_2s_6s_3)$ to $(s_2s_3)$. It then notifies $s_6$ about
the update of $\Flow_B$, allowing $s_6$ to safely remove the forwarding state
corresponding to this flow.

Once notified by $s_2$, $s_6$ infers that link $\Link{6}{3}$ now has available
capacity to carry $\Flow_R$. So, it installs the corresponding forwarding state
and notifies $s_2$, which is the upstream switch on the new path of $\Flow_R$.
Then, $s_2$ infers that the new downstream switch is ready and that
$\Link{2}{6}$ has enough capacity, so it moves $\Flow_R$ to its new 
path.

Similarly, $s_3$ updates its forwarding state for $\Flow_R$ to flow on
$\Link{3}{4}$, notifying $s_7$ about the update. Meanwhile, switch $s_7$ infers
that link $\Link{7}{4}$ has enough capacity to carry $\Flow_B$; then, it
installs forwarding state for $\Flow_B$ and notifies $s_3$. Then $s_3$ moves
$\Flow_B$ onto $\Link{3}{7}$.

Notice that several update operations can run in parallel at multiple switches.
However, whenever operations have unsatisfied dependencies, switches must
coordinate. In this example, the longest dependency chain involves the three
operations that must occur in sequence at $s_2$, $s_6$, and $s_2$ again.
Therefore, the above execution accumulates the delay for the initial message
from the controller to arrive at the switches plus a round-trip delay between
$s_2$ and $s_6$. In contrast, if a centralized controller performed the update
following the same schedule, the update time would be affected by the sum of
three round-trip delays (two to $s_2$ and one to $s_6$). We aim to replace the
expensive communication (in terms of latency) between the switches and the
controller with simple and fast coordination messages among neighboring
switches.

\subsection{Dealing with deadlocks}

As in previous work on network update problems, it is worth asking: is it always
possible to reach the target configuration while preserving all the desired
consistency properties and resource constraints during the transition?
Unfortunately, similarly to the centralized case~\cite{Jin+DSN+2014}, some
updates are not feasible: that is, even if the target configuration do not
violate any of the three properties, there exists no ordering of update
operations to reach the target. For example, in
Figure~\ref{fig.big-network-segment-deadlock}, moving first either $\Flow_B$ or
$\Flow_R$ creates congestion.

Assume that a feasible update ordering exists. Then, will a decentralized
approach always be able to find it? Unfortunately, without computing a possible
schedule in advance~\cite{McClurg.PLDI15}, inappropriate ordering can lead to
deadlocks where no further progress can be made. However, as discussed
in~\cite{Jin+DSN+2014}, computing a feasible schedule is a computationally 
hard problem.

In practice, even in centralized settings, the current state-of-the-art
approach, Dionysus~\cite{Jin+DSN+2014}, cannot entirely avoid deadlocks. 
In such cases, Dionysus reduces flow rates to continue an
update without violating the consistency properties; however, this comes at 
the cost of lower network throughput.

Deadlocks pose an important challenge for us as a decentralized approach is
potentially more likely to enter deadlocks due to the lack of global
information. To deal with deadlocks, we develop two techniques that avoid
reducing network throughput (without resorting to rate-limit techniques~\cite{Jin+DSN+2014}
or reserving some capacities for the update operations~\cite{hong13}).

Our first technique is \emph{splitting volume}, which divides a flow's traffic
onto its old and new paths to resolve a deadlock. The second technique is
\emph{segmentation}, which allows to update different flow ``segments''
independently of one another. As we discussed, segmentation allows an update to
complete faster via parallelization. In this case, it also helps to resolve
certain deadlocks. Before presenting our techniques in detail
(\S\ref{sec:algo}), we illustrate them with intuitive examples.

\begin{figure}[ptb]
	\centering
		\scalebox{0.8}{
			\begin{tikzpicture}[scale=1.6]
			\label{fig.notation-flow-2}			
			\node[above] at (0,0.8) {$F_R{:}5$};
			\draw[red!75!white,thick] plot [smooth,tension=0.5] coordinates { (0.25, 0.95) (0.6, 0.95) } [arrow inside={end=stealth,opt={red!75!white,scale=2}}{0.99}];
			
			\node[above] at (0,0.5) {$F_G{:}5$};
			\draw[green!50!black,thick] plot [smooth,tension=0.5] coordinates { (0.25, 0.65) (0.6, 0.65) } [arrow inside={end=stealth,opt={green!50!black,scale=2}}{0.99}];
			
			\node[above] at (0,0.2) {$F_B{:}5$};
			\draw[blue,thick] plot [smooth,tension=0.5] coordinates { (0.25, 0.35) (0.6, 0.35) } [arrow inside={end=stealth,opt={blue,scale=2}}{0.99}];
			\end{tikzpicture}
		}~
		\subfloat[\scriptsize{Configuration $\NetworkConfig_2$}]{
			\label{fig.c2-segment-deadlock}
			\scalebox{0.95}{
				\begin{tikzpicture}[scale=0.75]
				\node [vertex] (1) at (-0.5,0) {$s_1$};
				\node [vertex] (2) at (0.5,0) {$s_2$};
				\node [vertex] (3) at (1.5,0) {$s_3$};
				\node [vertex] (4) at (2.5,0) {$s_4$};
				\node [vertex] (5) at (3.5,0) {$s_5$};
				\node [vertex] (6) at (1,1) {$s_6$};
				\node [vertex] (7) at (2,1) {$s_7$};
				\draw[act] (1) to (2);
				\draw[act] (2) to (3);
				\draw[act] (2) to (6);
				\draw[act] (3) to (4);
				\draw[act] (3) to (6);
				\draw[act] (3) to (7);
				\draw[act] (4) to (5);
				\draw[act] (4) to (7);
				
				\draw[green!40!black, thick] plot [smooth,tension=0.5] coordinates { (0.5, -0.2) (1.5, -0.2) (2.5, -0.2) } [arrow inside={end=stealth,opt={green!40!black,scale=1.5}}{0.99}];
				
				\draw[blue, thick] plot [smooth,tension=0.5] coordinates { (-0.4,0.2) (0.5,0.2) (0.8,1.1) (1.15,1.15) (1.5,0.2) (2.5,0.2) } [arrow inside={end=stealth,opt={blue,scale=2}}{0.99}];
				
				\draw[red!75!white, thick] plot [smooth,tension=0.5]  coordinates { (0.65,0.15) (1.5,0.2) (1.8,1.1) (2.15,1.15) (2.5,0.2) (3.5, 0.2)} [arrow inside={end=stealth,opt={red!50!white,scale=2}}{0.99}];
				\end{tikzpicture}
			}
		}
		\subfloat[\scriptsize{Configuration $\NetworkConfig_2'$}]{
			\label{fig.c3-segment-deadlock}
			\scalebox{0.95}{
				\begin{tikzpicture}[scale=0.75]
				\node [vertex] (1) at (-0.5,0) {$s_1$};
				\node [vertex] (2) at (0.5,0) {$s_2$};
				\node [vertex] (3) at (1.5,0) {$s_3$};
				\node [vertex] (4) at (2.5,0) {$s_4$};
				\node [vertex] (5) at (3.5,0) {$s_5$};
				\node [vertex] (6) at (1,1) {$s_6$};
				\node [vertex] (7) at (2,1) {$s_7$};
				\draw[act] (1) to (2);
				\draw[act] (2) to (3);
				\draw[act] (2) to (6);
				\draw[act] (3) to (4);
				\draw[act] (3) to (6);
				\draw[act] (3) to (7);
				\draw[act] (4) to (5);
				\draw[act] (4) to (7);
				
				\draw[green!40!black, thick] plot [smooth,tension=0.5] coordinates { (0.5, -0.2) (1.5, -0.2) (2.5, -0.2) } [arrow inside={end=stealth,opt={green!40!black,scale=1.5}}{0.99}];
				
				\draw[blue, thick] plot [smooth,tension=0.5] coordinates { (-0.4,0.2) (1.45,0.25) (1.8,1.1) (2.15,1.15) (2.5,0.2) } [arrow inside={end=stealth,opt={blue,scale=2}}{0.99}];
				
				\draw[red!75!white, thick] plot [smooth,tension=0.5]  coordinates { (0.5,0.2) (0.8,1.1) (1.15,1.15) (1.5,0.25) (2.5, 0.2) (3.5, 0.2)} [arrow inside={end=stealth,opt={red!50!white,scale=2}}{0.99}];
				\end{tikzpicture}
			}
		}
	\caption{An update with segment deadlock.}
	\label{fig.big-network-segment-deadlock}
\end{figure}
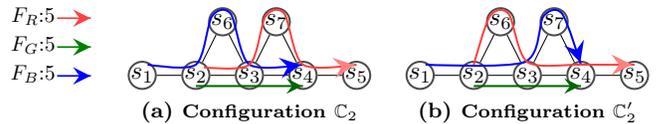
\smartparagraph{Segmentation example.}
Consider the example in Figure~\ref{fig.big-network-segment-deadlock},
where every flow has size 5. This case presents a deadlock that prevents flows
from being updated all at once. In particular, if we first move $\Flow_R$, link
$\Link{3}{4}$ becomes congested. Similarly, if we first move $\Flow_B$, link
$\Link{2}{3}$ is then congested.

We resolve this deadlock by segmenting these flows as:
$\Flow_R{=}\{\NwSwitch_2{\ldots}\NwSwitch_3, \NwSwitch_3{\ldots}\NwSwitch_4,
\NwSwitch_4{\ldots}\NwSwitch_5\}$, $\Flow_B{=}\{\NwSwitch_1{\ldots}\NwSwitch_2,
\NwSwitch_2{\ldots}\NwSwitch_3,$ $\NwSwitch_3{\ldots}\NwSwitch_4\}$.
Then, switches $\NwSwitch_2$ and $\NwSwitch_6$ coordinate to first move segment
$\{\NwSwitch_2{\ldots}\NwSwitch_3\}$ of $\Flow_R$ followed by the same segment of
$\Flow_B$. Independently, switches $\NwSwitch_3$ and $\NwSwitch_7$ move segment
$\{\NwSwitch_3{\ldots}\NwSwitch_4\}$ of $\Flow_B$ and $\Flow_R$, in this order.

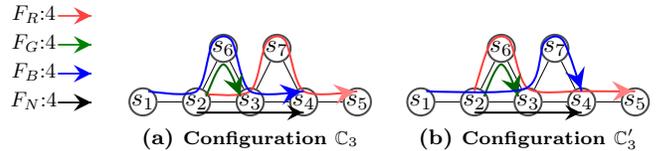
\begin{figure}[ptb]
	\centering
	\scalebox{0.8}{
		\begin{tikzpicture}[scale=1.6]
		\label{fig.notation-flow-3}
		\node[above] at (0,0.8) {$F_R{:}4$};
		\draw[red!75!white,thick] plot [smooth,tension=0.5] coordinates { (0.25, 0.95) (0.6, 0.95) } [arrow inside={end=stealth,opt={red!75!white,scale=2}}{0.99}];
		
		\node[above] at (0,0.5) {$F_G{:}4$};
		\draw[green!50!black,thick] plot [smooth,tension=0.5] coordinates { (0.25, 0.65) (0.6, 0.65) } [arrow inside={end=stealth,opt={green!50!black,scale=2}}{0.99}];
		
		\node[above] at (0,0.2) {$F_B{:}4$};
		\draw[blue,thick] plot [smooth,tension=0.5] coordinates { (0.25, 0.35) (0.6, 0.35) } [arrow inside={end=stealth,opt={blue,scale=2}}{0.99}];
		
		\node[above] at (0,-0.1) {$F_N{:}4$};
		\draw[black,thick] plot [smooth,tension=0.5]  coordinates { (0.25, 0.05) (0.6, 0.05)} [arrow inside={end=stealth,opt={black,scale=2}}{0.99}];
		\end{tikzpicture}
	}~
		\subfloat[\scriptsize{Configuration $\NetworkConfig_3$}]{
			\label{fig.c1-splittable-deadlock}
			\scalebox{0.95}{
				\begin{tikzpicture}[scale=0.75]
					\node [vertex] (1) at (-0.5,0) {$s_1$};
					\node [vertex] (2) at (0.5,0) {$s_2$};
					\node [vertex] (3) at (1.5,0) {$s_3$};
					\node [vertex] (4) at (2.5,0) {$s_4$};
					\node [vertex] (5) at (3.5,0) {$s_5$};
					\node [vertex] (6) at (1,1) {$s_6$};
					\node [vertex] (7) at (2,1) {$s_7$};
					\draw[act] (1) to (2);
					\draw[act] (2) to (3);
					\draw[act] (2) to (6);
					\draw[act] (3) to (4);
					\draw[act] (3) to (6);
					\draw[act] (3) to (7);
					\draw[act] (4) to (5);
					\draw[act] (4) to (7);

					\draw[red!75!white, thick] plot [smooth,tension=0.5]  coordinates { (0.65,0.15) (1.5,0.2) (1.8,1.1) (2.15,1.15) (2.5,0.2) (3.5, 0.2)} [arrow inside={end=stealth,opt={red!50!white,scale=2}}{0.99}];
					
					\draw[blue, thick] plot [smooth,tension=0.5] coordinates { (-0.4,0.2) (0.5,0.2) (0.8,1.1) (1.15,1.15) (1.5,0.2) (2.5,0.2) } [arrow inside={end=stealth,opt={blue,scale=2}}{0.99}];

					\draw[green!40!black, thick] plot [smooth,tension=0.6] coordinates { (0.68,0.1) (1,0.7) (1.32,0.1) } [arrow inside={end=stealth,opt={green!40!black,scale=2}}{0.99}];

					\draw[black, thick] plot [smooth,tension=0.5] coordinates { (0.5, -0.2) (1.5, -0.2) (2.5, -0.2) } [arrow inside={end=stealth,opt={black,scale=1.5}}{0.99}];
				\end{tikzpicture}
			}
		}
		\subfloat[\scriptsize{Configuration $\NetworkConfig_3'$}]{
			\label{fig.c2-prime-splittable-deadlock}
			\scalebox{0.95}{
				\begin{tikzpicture}[scale=0.75]
					\node [vertex] (1) at (-0.5,0) {$s_1$};
					\node [vertex] (2) at (0.5,0) {$s_2$};
					\node [vertex] (3) at (1.5,0) {$s_3$};
					\node [vertex] (4) at (2.5,0) {$s_4$};
					\node [vertex] (5) at (3.5,0) {$s_5$};
					\node [vertex] (6) at (1,1) {$s_6$};
					\node [vertex] (7) at (2,1) {$s_7$};
					\draw[act] (1) to (2);
					\draw[act] (2) to (3);
					\draw[act] (2) to (6);
					\draw[act] (3) to (4);
					\draw[act] (3) to (6);
					\draw[act] (3) to (7);
					\draw[act] (4) to (5);
					\draw[act] (4) to (7);
					
					\draw[red!75!white, thick] plot [smooth,tension=0.5]  coordinates { (0.5,0.2) (0.8,1.1) (1.15,1.15) (1.5,0.25) (2.5, 0.2) (3.5, 0.2)} [arrow inside={end=stealth,opt={red!50!white,scale=2}}{0.99}];
					
					\draw[green!40!black, thick] plot [smooth,tension=0.6] coordinates { (0.68,0.1) (1,0.7) (1.32,0.1) } [arrow inside={end=stealth,opt={green!40!black,scale=2}}{0.99}];
					
					\draw[black, thick] plot [smooth,tension=0.5] coordinates { (0.5, -0.2) (1.5, -0.2) (2.5, -0.2) } [arrow inside={end=stealth,opt={black,scale=1.5}}{0.99}];

					\draw[blue, thick] plot [smooth,tension=0.5] coordinates { (-0.4,0.2) (1.45,0.25) (1.8,1.1) (2.15,1.15) (2.5,0.2) } [arrow inside={end=stealth,opt={blue,scale=2}}{0.99}];

				\end{tikzpicture}
			}
		}
	\caption{An update with splittable deadlock.}
	\label{fig.splittable-deadlock}
\end{figure}
\smartparagraph{Splitting volume example.}
Consider the example in Figure~\ref{fig.splittable-deadlock}, where every flow
has size 4. This case presents a deadlock because we cannot move $\Flow_R$ first
without congesting $\Link{2}{6}$ or move $\Flow_B$ first without congesting
$\Link{2}{3}$.

We resolve this deadlock by splitting the flows. Switch $\NwSwitch_2$ infers
that $\Link{2}{3}$ has 2 units of capacity and starts moving the corresponding
fraction of $\Flow_B$ onto that link. This movement releases sufficient capacity
to move $\Flow_R$ to $\Link{2}{6}$ and $\Link{6}{3}$. Once $\Flow_R$ is moved,
there is enough capacity to complete the mode of $\Flow_B$.

Note that, we could even speed up part of the update by moving two units of
$\Flow_R$ simultaneously to moving two units of $\Flow_B$ at the very beginning.
This is what our decentralized solution would do.
\section{ez-Segway}
\label{sec:algo}

\begin{figure*}[tb]
\centering
\scalebox{0.85}{
	\subfloat[No reversed order common switches]{
		~~~~~
		\begin{tikzpicture}[scale=1.24]\label{fig:no-cycle}
		\draw[black,thick] plot [smooth,tension=0.5] coordinates { (0,0.5) (0.5,0) (1,0.5) (1.5,0) (2,0.4) (2.85,0.45) } [arrow inside={end=stealth,opt={black,scale=2}}{0.999}];
		
		\draw[dashed,thick] plot [smooth,tension=0.5] coordinates { (0,0.5) (0.5,1) (1,0.5) (1.5,1) (2,0.6) (2.85,0.55)} [arrow inside={end=stealth,opt={black,scale=2}}{0.999}];
		
		\foreach \x in {0,...,3}
		{
			\FPeval{\result}{\x};
			\FPtrunc{\val}{\result}{0}
			\node [vertex, draw=black, fill=white!70!blue, minimum size=5pt,inner sep=1pt] (\x) at (\x,0.5) {$s_\val$};
		}
		
		\foreach \x in {0,...,1}
		{
			\FPeval{\result}{\x+4};
			\FPtrunc{\val}{\result}{0}
			\node [vertex, minimum size=5pt,inner sep=1pt] (\x) at (\x + 0.5,0) {$s_\val$};
		}
		\foreach \x in {6,...,7}
		{
			\node [vertex, minimum size=5pt,inner sep=1pt] (\x) at (\x - 5.5,1) {$s_\x$};
		}
		
		\end{tikzpicture}
		~~~~~
	}
}~
\scalebox{0.85}{
	\subfloat[A pair of reversed order common switches]{
		~~~~~
		\begin{tikzpicture}[scale=1.24]\label{fig:one-cycle}
		\draw[black,thick] plot [smooth,tension=1] coordinates { (0,0.5) (0.5,0) (1,0.5) (1.5,0) (2,0.5) (2.5,0) (3,0.4) (3.8,0.45)} [arrow inside={end=stealth,opt={black,scale=2}}{0.999}];
		
		\draw[black,thick,dashed] plot [smooth,tension=1] coordinates { (0,0.5) (1.25,1.3) (2.5,1) (2.1,0.5) (1,0.5) (1.5,1) (2.9,0.6) (3.8,0.55)} [arrow inside={end=stealth,opt={dashed,scale=2}}{0.999}];
		
		\foreach \x in {0,3}
		{
			\FPeval{\result}{\x};
			\FPtrunc{\val}{\result}{0}
			\node [vertex, draw=black, fill=white!70!blue, minimum size=5pt,inner sep=1pt] (\x) at (\x,0.5) {$s_\val$};
		}
		
		\foreach \x in {1,2}
		{
			\FPeval{\result}{\x};
			\FPtrunc{\val}{\result}{0}
			\node [vertex, draw=black, fill=white!70!red, minimum size=5pt,inner sep=1pt] (\x) at (\x,0.5) {$s_\val$};
		}
		
		\node [vertex, draw=black, fill=white!70!blue, minimum size=5pt,inner sep=1pt] (9) at (4,0.5) {$s_9$};
		
		\foreach \x in {0,...,2}
		{
			\FPeval{\result}{\x+4};
			\FPtrunc{\val}{\result}{0}
			\node [vertex, minimum size=5pt,inner sep=1pt] (\x) at (\x + 0.5,0) {$s_\val$};
		}
		
		\node [vertex, minimum size=5pt,inner sep=1pt] (7) at (1.5,1) {$s_7$};
		\node [vertex, minimum size=5pt,inner sep=1pt] (8) at (2.5,1) {$s_8$};
		\end{tikzpicture}~~~
	}
}~
\scalebox{0.85}{
	\subfloat[Two pairs of reversed common switches]{
		\begin{tikzpicture}[scale=1.18]\label{fig:two-cycle}
	
		\draw[black,thick] plot [smooth] coordinates { (0,0.5) (1,0.5) (2,0.5) (3,0.5) (4,0.5) (5,0.5) (5.8,0.5)} [arrow inside={end=stealth,opt={black,scale=2}}{0.999}];
		
		\draw[dashed,thick] plot [smooth,tension=0.75] coordinates { (0,0.5) (1.5,-0.35) (3,0.5) (2.6,0.9) (2,0.5) (1.5,0.1) (1,0.5) (1.6,1.02) (3,1.2) (4.3,1.02) (5,0.5) (4.5,0) (4,0.5) (5,1) (5.9,0.6)} [arrow inside={end=stealth,opt={dashed,scale=2}}{0.999}];
		
		\foreach \x in {0,2,6}
		{
			\FPeval{\result}{\x};
			\FPtrunc{\val}{\result}{0}
			\node [vertex, draw=black, fill=white!70!blue, minimum size=5pt,inner sep=1.5pt] (\x) at (\x,0.5) {$s_\val$};
		}
	
		\foreach \x in {1,3,4,5}
		{
			\FPeval{\result}{\x};
			\FPtrunc{\val}{\result}{0}
			\node [vertex, draw=black, fill=white!70!red, minimum size=5pt,inner sep=1.5pt] (\x) at (\x,0.5) {$s_\val$};
		}
		
		\node [vertex, draw=black, minimum size=5pt,inner sep=1.2pt] (7) at (1.5,-0.35) {$s_7$};
		\node [vertex, draw=black, minimum size=5pt,inner sep=1.2pt] (8) at (1.5,0.1) {$s_8$};
		\node [vertex, draw=black, minimum size=5pt,inner sep=1.2pt] (9) at (3,1.2) {$s_9$};
		\node [vertex, draw=black] (10) at (4.5,0) {$s_{10}$};
		
		\end{tikzpicture}
	}
}
\vspace{-1ex}
\caption{Different cases of segmentation.}
\label{fig:loop}
\end{figure*}
  
\ezsegway is a mechanism to update the forwarding state in a fast and consistent
manner: it improves update completion time while preventing forwarding anomalies
(\ie, black-hole, loops), and avoiding the risk of link congestion.

To achieve fast updates, our design leverages a small 
number of crucial, yet simple, update
functionalities that are performed by the switches. The centralized controller 
leverages its global visibility into network conditions to
compute and transmit \textit{once} to the switches the information needed to 
execute an update. The switches then schedule and perform network update 
operations without interacting with the controller.

As per our problem formulation (\S\ref{sec:problem}), we assume the controller
has knowledge of the initial and final network configurations $\NetworkConfig,
\NetworkConfig'$. In practice, each network update might be determined by some
application running on the controller (e.g., routing, monitoring, etc.). In our
settings, the controller is centralized and we leave it for future work to
consider the case of distributed controllers.

\smartparagraph{Controller responsibilities. }In \ezsegway, on every 
update, the controller performs two tasks: $(i)$ The
controller identifies flow segments ({\em to enable parallelization}), and it
identifies execution dependencies among segments ({\em to guarantee
consistency}). These dependencies are computed using a fast heuristic that
classifies segments into two categories and establishes, for each segment pair,
which segment must precede another one. $(ii)$ The controller computes the
dependency graph ({\em to avoid congestion}), which encodes the dependencies
among update operations and the bandwidth requirements for the update operations.

\smartparagraph{Switch responsibilities. }The set of functions at switches 
encompasses simple message exchange among
adjacent switches and a greedy selection of update operations to be executed
based on the above information (\ie, the  dependencies among segments) provided by the controller. These functions are
computationally inexpensive and easy to implement in currently available
programmable switches, as we show in Section~\ref{sec:micro-benchmark}.

Next, we first describe in more detail our segmentation technique, which
\ezsegway uses to speed up network updates while guaranteeing anomaly-free
forwarding during updates. Then, we describe our scheduling mechanism based on
segments of flows, which avoids congestion. Finally, we show how to cast this
mechanism into a simple distributed setting that requires minimal computational
power in the switches.

\subsection{Flow segmentation}\label{sec:reliability-segmentation}
Our segmentation technique provides two benefits: it speeds up the update
completion time of a flow to its new path and it reduces the risk of update
deadlocks due to congested links by allowing a more fine-grained control of the
flow update. Segment identification is \emph{performed by the controller} when a
network state update is triggered. We first introduce the idea behind
segmentation and then describe our technique in details.
 
 
\smartparagraph{Update operation in the distributed approach.} Consider the update problem represented in
Figure~\ref{fig:no-cycle}, where a flow $F$ needs to be moved from the old
(solid line) path $(s_0s_4\dots s_2s_3)$ to the new (dashed-line) path $(s_0s_6\dots s_2s_3)$.
A simple approach would work as follows. A message is sent from $s_3$ to its
predecessor on the new path (\ie, $s_2$) for acknowledging the beginning of the
flow update. Then, every switch that receives this message forwards it as soon
as it installs the new rule for forwarding packets on the new path. Since the
message travels in the reverse direction of the new path, the reception of such
message guarantees that each switch on the downstream path consistently updated
its forwarding table to the new path, thus preventing any risk of black-holes or
forwarding loops anomalies. Once the first switch of the new path (\ie, $s_0$)
switches to the new path, a new message is sent from $s_0$ towards $s_3$ along
the old path for acknowledging that no packets will be forwarded anymore along
the old path, which can be therefore safely removed. Every switch that receives
this new message removes the old forwarding entry of the old path and afterwards
forwards it to its successor on the old path. We call this flow update technique
from an old path to the new one \emph{\basicmigration}. It is easy to observe
that \basicmigration prevents forwarding loops and black-holes 
\if \TR 1
(proof in Appendix~\ref{sec:correctness}). 
\else
(proof in the extended version of this paper~\cite{ez-segway-tr}). 
\fi

\newcommand{ \basicMigrationCorrectnessStatement}{ \basicmigration is black-hole- and forwarding-loop-free.}
\begin{theorem}\label{theo:basic-migration}
 \basicMigrationCorrectnessStatement
\end{theorem} 

 
\smartparagraph{Speeding up a flow update with segmentation. } It can be observed that the whole flow
update operation could be performed faster. With segmentation the subpath
$(s_0s_4s_1)$ of the old path can be updated with \basicmigration to
$(s_0s_6s_1)$ while subpath $(s_1s_5s_2)$ of the old path is updated
with \basicmigration to $(s_1s_7s_2)$. In fact, $s_6$ does not have to wait for
$s_1$ to update its path since $s_1$ is always guaranteed to have a forwarding
path towards $s_3$.  We say that the pairs $(s_0s_4s_1,s_0s_6s_1)$ and
$(s_1s_5s_2,s_1s_7s_2)$ are two ``segments'' of the flow update. Namely, a
\emph{segment} of a flow update from an old path $P_o$ to a new path $P_n$ is a
pair of subpaths $P'_o$ and $P'_n$ of $P_o$ and $P_n$, respectively, that start
with the same first switch. We denote a segment by $(P'_o,P'_n)$. The action of
\emph{updating a segment} consists in performing \basicmigration from $P'_o$ to
$P'_n$.


\smartparagraph{Dependencies among segments.} In some cases, finding a set of
segments that can be updated in parallel is not possible. For instance, in
Figure~\ref{fig:one-cycle}, we need to guarantee that $s_2$ updates its next-hop
switch only after $s_1$ has updated its own next-hop; otherwise, a forwarding
loop along $(s_2,s_1,s_5)$ arises. While one option would be to update the whole flow
using \basicmigration along the reversed new path, some parallelization is still
possible. Indeed, we can update segments $S_1{=}(s_0s_4s_1,s_0s_8s_2)$ and
$S_2{=}(s_1s_5s_2,s_1s_7s_3)$ in parallel without any risk of creating a
forwarding loop. In fact, $s_2$, which is a critical switch, is not yet updated
and the risk of forwarding oscillation is avoided. After $S_2$ is updated, also
segment $S_3{=}(s_2s_6s_3,s_2s_1)$ can be updated. Every time two switches appear
in reversed order in the old and new path, one of the two switches has to wait
until the other switch completes its update.
 

\smartparagraph{Anomaly-free segment update heuristic.} It would be tempting to
create as many segments as possible so that the update operation could benefit
at most from parallelization. However, if the chain of dependencies among the segments is too long the updates may be unnecessary slowed down. In practice, computing a maximal set of segments that minimize the chain of dependencies among them is not an easy task. We rely on a heuristic that classifies segments
into two categories called {\sc InLoop} and {\sc NotInLoop} and a dependency
mapping $dep: $ {\sc InLoop}${\rightarrow}${\sc NotInLoop} that assigns each
{\sc InLoop} segment  to a {\sc NotInLoop} segment that must be executed before
its corresponding {\sc InLoop} segment to avoid a forwarding loop. This
guarantees that the longest chain of dependencies is limited to at most $2$.
We call such technique \emph{\segmentedmigration}.
  
\newcommand{\segmentationHeuristicDescritption}{
Our heuristic works as follows. It identifies the set of common switches
$\NwSwitchSet_C$ among the old and new path, denotes by $\NwPairSwitchSet$ the
pairs of switches in $\NwSwitchSet_C$ that appear in reversed order in the two
paths, and sets $\NwSwitchSet_R$ to be the set of switches that appear in $\NwSwitchSet_C$. It selects a subset $\NwPairSwitchSet_R$ of $\NwPairSwitchSet$ that has the
following properties: 
$(i)$ for each two pairs of switches $(r,s)$ and $(r',s')$ in  $\NwPairSwitchSet_R$
neither $r'$ nor $s'$ are contained in the subpath from $r$ to $s$ in the old
path, 
$(ii)$ every switch belonging to a pair in $\NwSwitchSet_C$ is contained in at least a subpath
from $r$ to $s$ for a pair $(r,s)$ of $\NwPairSwitchSet_R$, and
$(iii)$ the number of pairs in $\NwPairSwitchSet_R$ is minimized.
Intuitively, each pair $(s_1,s_2)$ of $\NwPairSwitchSet_R$ represents a
pair of switches that may cause a forwarding loop unless $s_2$ is updated after
$s_1$. The complexity for computing these pairs of switches is $O(N)$, where $N$
is the number of switches in the network.

The segments are then created as follows. Let $\NwSwitchSet_R^1$ ($\NwSwitchSet_R^2$) 
be the set of switches that appear as the first 
(second) element in at least one pair of $\NwPairSwitchSet_R$.
Let $\NwSwitchSet_C^* \subseteq \NwSwitchSet_C$ be the set of common vertices
that are  contained in neither $\NwSwitchSet_R^1$ nor $\NwSwitchSet_R^2$.
For each switch $s$ in $\NwSwitchSet_C^* \cup \NwSwitchSet_R^1$ ($ \NwSwitchSet_R^2$), 
we create a {\sc NotInLoop} ({\sc InLoop}) segment $S{=}(P_o,P_n)$, where  $P_o$ ($P_n$)
is the subpath of the old (new) path from $s$ to the next switch $r$ in 
$\NwSwitchSet_C^* \cup \NwSwitchSet_R^1 \cup \NwSwitchSet_R^2$. 
Morevoer, if $S$ is an
{\sc InLoop} segment, we set $dep(S)$ to be $S_r$, where $S_r$ is the segment
starting from $r$.
}
  
\segmentationHeuristicDescritption
The following theorem 
\if \TR 0
(proof in the extended version of this paper~\cite{ez-segway-tr})
\else
(proof in Appendix~\ref{sec:correctness}) 
\fi
 guarantees
 that \segmentedmigration does not create black-holes or
forwarding loops.

\newcommand{ \segmentedMigrationCorrectnessStatement}{  \segmentedmigration is black-hole- and forwarding-loop-free.}
\begin{theorem}\label{theo:segmentation-anomaly-freeness}
	\segmentedMigrationCorrectnessStatement
\end{theorem} 
As an example, consider the update depicted in Figure~\ref{fig:two-cycle}. 
All the switches in common between the old and new path are colored in red
or blue.
Let $\NwPairSwitchSet_R$ be $\{(s_1,s_3),(s_4,s_5)\}$ (colored in red) for which both properties $(i)$, $(ii)$, and
$(iii)$ hold. The resulting {\sc NotInLoop} segments are
$S_1{=}(s_1s_2s_3,s_1s_9s_5),S_2{=}(s_4s_5,s_4s_6), S_3{=}(s_0s_1,s_0s_7s_3)$, the
{\sc InLoop} segments are $S_4{=}(s_3s_4,s_3s_2s_8s_1),S_5=(s_5s_6,s_5s_{10}s_4)$,
while the dependencies $dep(S_4)$ and $dep(S_5)$ are $S_1$ and $S_2$, respectively.
   
\smartparagraph{Segmentation and middleboxes.} If the network policy dictates
that a specific flow of data traffic must traverse one or more middleboxes,
segmenting a flow must be carefully performed as the combination of fragments
from the old and the new path may not traverse the set of middleboxes. In these
cases, segmentation can only be applied if there are no {\sc InLoop} segments.

\subsection{Scheduling updates}
\label{sec:ez-segway-scheduling}

Performing a network update in a distributed manner requires coordination among
the switches. As we already showed in
Section~\ref{sec:reliability-segmentation}, segments' updates that can cause
forwarding loops must be executed according to a partial order of all the update
operations.

We now consider the problem of performing a network update without congesting
any link in the network. The main challenge of this procedure is how to avoid
deadlock states in which any update operation would create congestion. We
propose a heuristic that effectively reduces the risk of deadlocks, which we
positively experimented in our evaluation section
(Section~\ref{sec:evaluation}).
Our heuristic is correct, i.e., our heuristic finds a 
congestion-free migration only if the network update problem admits a congestion-free solution.

Figure~\ref{fig.select-update} shows a case in which an inappropriate schedule
of network operations leads to a deadlock situation. $\Flow_R$ and $\Flow_B$
have size $4$ while $\Flow_G$ and $\Flow_N$ have size $3$. The only operations
that can be performed without congesting a link are updating either  $\Flow_R$
or $\Flow_G$. Which one to choose first? Choosing $\Flow_G$ will not allow to
update $\Flow_B$ since $\Flow_G$ releases only 3 units of free capacity to
$\Link{2}{3}$. This leads to a deadlock that can only be solved via splitting
mechanisms, which increases the completion time of the update. Instead, choosing
to update $\Flow_R$ releases 4 units of capacity to $\Link{2}{3}$, which allows
us to update $\Flow_B$. In turns, by updating $\Flow_B$ we have enough capacity
to move $\Flow_G$ to its new path, and the update completes successfully.

In the previous example, we noted that splitting volumes helps avoiding deadlocks; however, this is less preferable than finding a sequence of update operations that do not
split flows for two main reasons: $(i)$
splitting flows requires more time to complete an update and $(ii)$ the physical
hardware may lack support for splitting flows based on arbitrary weights. In the
last case, we stress the fact that \ezsegway can be configured to disable flow
splitting during network updates.

Considering again the previous example, we observe that there are two types of
deadlocks we must consider. If a network update execution reaches a state in
which it is not possible to make any progress unless by splitting a flow, we say
that there is \emph{splittable deadlock} in the system. Otherwise, if even
splitting flows cannot make any progress, we say that there is an
\emph{unsplittable deadlock} in the network.

Even from a centralized perspective, computing a congestion-free migration is
not an easy task~\cite{Brandt2016On}. Our approach employs a mechanism to
centrally pre-compute a \emph{dependency graph} between links and segment
updates that can be used to infer which updates are more critical than others
for avoiding deadlocks.
 
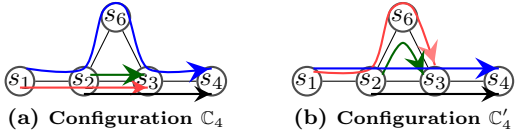
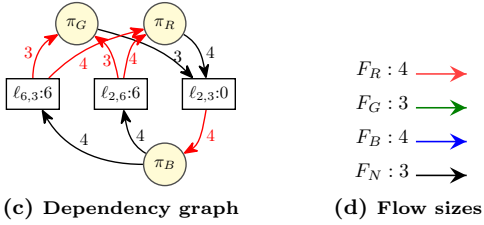
\begin{figure}[ptb]
	\centering
	\vspace{-1ex}
	\begin{tabular}{cc}
		\subfloat[\scriptsize{Configuration $\NetworkConfig_4$}]{
			\label{fig.c1-select-update}
			\scalebox{1}{
				\begin{tikzpicture}[scale=0.85]
				\node [vertex] (1) at (-0.5,0) {$s_1$};
				\node [vertex] (2) at (0.5,0) {$s_2$};
				\node [vertex] (3) at (1.5,0) {$s_3$};
				\node [vertex] (4) at (2.5,0) {$s_4$};
				\node [vertex] (6) at (1,1) {$s_6$};
				\draw[act] (1) to (2);
				\draw[act] (2) to (3);
				\draw[act] (2) to (6);
				\draw[act] (3) to (4);
				\draw[act] (3) to (6);
				
				\draw[red!75!white, thick, thick] plot [smooth,tension=0.5] coordinates { (-0.5, -0.1) (0.5, -0.1) (1.5, -0.1) } [arrow inside={end=stealth,opt={red!75!white, thick,scale=1.5}}{0.99}];
				
				\draw[blue, thick] plot [smooth,tension=0.5] coordinates { (-0.4,0.2) (0.5,0.2) (0.8,1.1) (1.15,1.15) (1.5,0.2) (2.5,0.2) } [arrow inside={end=stealth,opt={blue,scale=2}}{0.99}];
				
				\draw[green!40!black, thick] plot [smooth,tension=0.6] coordinates { (0.6, 0.1) (1.45, 0.1) } [arrow inside={end=stealth,opt={green!40!black,scale=2}}{0.99}];
				
				\draw[black, thick] plot [smooth,tension=0.5] coordinates { (0.5, -0.2) (1.5, -0.2) (2.5, -0.2) } [arrow inside={end=stealth,opt={black,scale=1.5}}{0.99}];
				\end{tikzpicture}
			}
		}&
		\subfloat[\scriptsize{Configuration $\NetworkConfig_4'$}]{
			\label{fig.c2-select-update}
			\scalebox{1}{
				\begin{tikzpicture}[scale=0.85]
				\node [vertex] (1) at (-0.5,0) {$s_1$};
				\node [vertex] (2) at (0.5,0) {$s_2$};
				\node [vertex] (3) at (1.5,0) {$s_3$};
				\node [vertex] (4) at (2.5,0) {$s_4$};
				\node [vertex] (6) at (1,1) {$s_6$};
				\draw[act] (1) to (2);
				\draw[act] (2) to (3);
				\draw[act] (2) to (6);
				\draw[act] (3) to (4);
				\draw[act] (3) to (6);
				
				\draw[red!75!white, thick] plot [smooth,tension=0.5]  coordinates { (-0.4, 0.15) (0.5,0.2) (0.8,1.1) (1.15,1.15) (1.5,0.25) } [arrow inside={end=stealth,opt={red!50!white,scale=2}}{0.99}];
				
				\draw[green!40!black, thick] plot [smooth,tension=0.6] coordinates { (0.68,0.1) (1,0.6) (1.32,0.1) } [arrow inside={end=stealth,opt={green!40!black,scale=2}}{0.99}];
				
				\draw[black, thick] plot [smooth,tension=0.5] coordinates { (0.5, -0.2) (1.5, -0.2) (2.5, -0.2) } [arrow inside={end=stealth,opt={black,scale=1.5}}{0.99}];
				
				\draw[blue, thick] plot [smooth,tension=0.5] coordinates { (-0.4,0.2) (1.5,0.2) (1.9,0.2) (2.5,0.2) } [arrow inside={end=stealth,opt={blue,scale=2}}{0.99}];
				
				\end{tikzpicture}
			}
		}
		\\
		\subfloat[\scriptsize{Dependency graph}]{
			\label{fig.c1-dep-graph}
			\scalebox{0.65}{
				\begin{tikzpicture}[scale=1.8, thick,
				StealthFill/.tip={Stealth[line width=1pt, scale=1.5]}, arrows={[round]}]
				
				\node [move] (TR) at (2,1.6) {$\MovePath_R$};
				\node [move] (TG) at (1,1.6) {$\MovePath_G$};
				\node [move] (TB) at (2,0) {$\MovePath_B$};

				\node [resource] (L63) at (0.5,0.8) {$\Link{6}{3}$:$6$};
				\node [resource] (L26) at (1.5,0.8) {$\Link{2}{6}$:$6$};
				\node [resource] (L23) at (2.5,0.8) {$\Link{2}{3}$:$0$};
				
				\draw[-{StealthFill[fill=black]}, black] (TR) to[out=-30,in=100] node[right, pos=0.5] {$4$} (L23);
				\draw[-{StealthFill[fill=black]}, black] (TG) to[out=-15,in=135] node[right, pos=0.7] {$3$} (L23);
				\draw[-{StealthFill[fill=black]}, black] (TB) to[out=150,in=-80] node[right, pos=0.5] {$4$} (L26);
				\draw[-{StealthFill[fill=black]}, black] (TB) to[out=180,in=-60] node[above, pos=0.5] {$4$} (L63);
				\draw[-{StealthFill[fill=red]}, red] (L23) to[out=-100,in=30] node[right, pos=0.5] {$4$} (TB);
				\draw[-{StealthFill[fill=red]}, red] (L26) to[out=80,in=210] node[below, pos=0.65] {$4$} (TR);
				\draw[-{StealthFill[fill=red]}, red] (L63) to[out=45,in=195] node[below, pos=0.4] {$4$} (TR);
				\draw[-{StealthFill[fill=red]}, red] (L26) to[out=100,in=-35] node[below, pos=0.65] {$3$} (TG);
				\draw[-{StealthFill[fill=red]}, red] (L63) to[out=90,in=210] node[left, pos=0.5] {$3$} (TG);
				\end{tikzpicture}
			}
		}&
		\subfloat[\scriptsize{Flow sizes}]{
			~~~
			\scalebox{0.75}{
				\begin{tikzpicture}[scale=2]
				\label{fig.notation-select-update}
				\node[above] at (0,0.8) {$F_R:4$};
				\draw[red!75!white,thick] plot [smooth,tension=0.5] coordinates { (0.3, 0.9) (0.75, 0.9) } [arrow inside={end=stealth,opt={red!75!white,scale=2}}{0.99}];
				
				\node[above] at (0,0.5) {$F_G:3$};
				\draw[green!50!black,thick] plot [smooth,tension=0.5] coordinates { (0.3, 0.6) (0.75, 0.6) } [arrow inside={end=stealth,opt={green!50!black,scale=2}}{0.99}];
				
				\node[above] at (0,0.2) {$F_B:4$};
				\draw[blue,thick] plot [smooth,tension=0.5] coordinates { (0.3, 0.3) (0.75, 0.3) } [arrow inside={end=stealth,opt={blue,scale=2}}{0.99}];
				
				\node[above] at (0,-0.1) {$F_N:3$};
				\draw[black,thick] plot [smooth,tension=0.5]  coordinates { (0.3, 0) (0.75, 0)} [arrow inside={end=stealth,opt={black,scale=2}}{0.99}];
				\end{tikzpicture}
			}
			~~~
		}
	\end{tabular}
	\caption{Choosing a correct update operation.}
	\label{fig.select-update}
\end{figure}

\noindent\textbf{The dependency graph.}
The dependency graph captures the complex set of dependencies between the
network update operations and the available link capacities in the network.
Given a pair of current and target configurations $\NetworkConfig,
\NetworkConfig'$, any execution of network operation $\MovePath$~(\ie,
updating a part of flow to its new path) requires some link capacity from the links
on the new path and releases some link capacity on the old path. We formalize
these dependencies in the \emph{dependency graph}, which is a bipartite graph
$\DependenceGraph{\MovePathSet}{\LinkSetDependent}{E_{free}}{E_{req}}$, where
the two subsets of vertices $\MovePathSet $ and $\LinkSetDependent$ represent
the {\em update operation set} and the {\em link set}, respectively. Each link
vertex $\Link{i}{j}$ is assigned a value representing its current available
capacity. Sets $E_{free}$ and $E_{req}$ are defined as follows:

	\noindent $\bullet$ $E_{free}$ is the set of directed edges from vertices in
	$\MovePathSet$ to vertices in $\LinkSetDependent$. A weighted edge with
	value $v$ from $\MovePath$ to a link $\Link{i}{j}$ represents the increase
	of available capacity of $v$ units at $\Link{i}{j}$ by performing
	$\MovePath$.\\
	$\bullet$ $E_{req}$ is the set of directed edges from vertices
	$\LinkSetDependent$ to vertices $\MovePathSet$. A weighted edge with value
	$v$ from link $\Link{i}{j}$ to $\MovePath$ represents the available
	capacity at $\Link{i}{j}$ that is needed to execute $\MovePath$.

Consider the example update of Figure~\ref{fig.select-update}(b) and (d), where
flows $\Flow_R$, $\Flow_G$, and $\Flow_B$ change paths. The corresponding
dependency graph is shown in Figure~\ref{fig.select-update}(c). In the graph,
circular nodes denote operations and rectangular nodes denote link capacities.
$E_{free}$ edges are shown in black and $E_{req}$ edges are shown in red; the
weight annotations reflect the amount of increased and requested available
capacity, respectively. For instance, $\MovePath_R$ requires 4 units of capacity
from $\Link{2}{6}$ and $\Link{6}{3}$, while it increases available capacity of
$\Link{2}{3}$ by the same units.

We now explore the dependency graph  with respect to the two techniques that
we introduced for solving deadlocks: segmentation and splitting flow volumes.

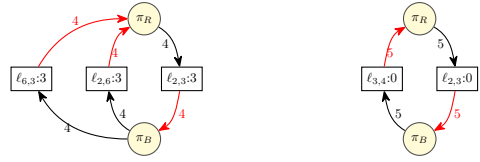
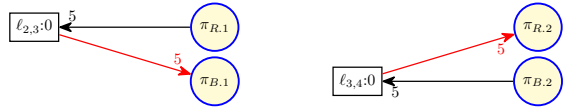
\begin{figure}[!t]
	\centering
	\subfloat[\scriptsize{Splittable deadlock}]{
		\label{fig.deadlock-split}
		\scalebox{0.5}{
			~~~~~~~~~~
			\begin{tikzpicture}[scale=2, thick,
			StealthFill/.tip={Stealth[line width=1pt, scale=1.5]}, arrows={[round]}]
			
			\node [move] (T1) at (2,1.6) {$\MovePath_R$};
			\node [move] (T3) at (2,0) {$\MovePath_B$};
			
			\node [resource] (L63) at (0.5,0.8) {$\Link{6}{3}$:$3$};
			\node [resource] (L26) at (1.5,0.8) {$\Link{2}{6}$:$3$};
			\node [resource] (L23) at (2.5,0.8) {$\Link{2}{3}$:$3$};
			
			\draw[-{StealthFill[fill=black]}, black] (T1) to[out=-30,in=100] node[left, pos=0.5] {$4$} (L23);
			\draw[-{StealthFill[fill=black]}, black] (T3) to[out=150,in=-80] node[right, pos=0.5] {$4$} (L26);
			\draw[-{StealthFill[fill=black]}, black] (T3) to[out=180,in=-60] node[left, pos=0.5] {$4$} (L63);
			\draw[-{StealthFill[fill=red]}, red] (L23) to[out=-100,in=30] node[right, pos=0.5] {$4$} (T3);
			\draw[-{StealthFill[fill=red]}, red] (L26) to[out=80,in=210] node[below, pos=0.5] {$4$} (T1);
			\draw[-{StealthFill[fill=red]}, red] (L63) to[out=60,in=180] node[above, pos=0.5] {$4$} (T1);
			\end{tikzpicture}
			~~~~~~~~~~
		}
	}
	\subfloat[\scriptsize{Deadlock ex. in Fig.~\ref{fig.big-network-segment-deadlock}}]{
		\label{fig.deadlock}
		\scalebox{0.5}{
			~~~~~~~~~~~~~~~~~~
			\begin{tikzpicture}[scale=2, thick,
			StealthFill/.tip={Stealth[line width=1pt, scale=1.5]}, arrows={[round]}]
			
			\node [move] (T1) at (2,1.6) {$\MovePath_R$};
			\node [move] (T3) at (2,0) {$\MovePath_B$};
			
			\node [resource] (L34) at (1.5,0.8) {$\Link{3}{4}$:$0$};
			\node [resource] (L23) at (2.5,0.8) {$\Link{2}{3}$:$0$};
			
			\draw[-{StealthFill[fill=black]}, black] (T1) to[out=-30,in=100] node[left, pos=0.5] {$5$} (L23);
			\draw[-{StealthFill[fill=black]}, black] (T3) to[out=150,in=-80] node[right, pos=0.5] {$5$} (L34);
			\draw[-{StealthFill[fill=red]}, red] (L23) to[out=-100,in=30] node[right, pos=0.5] {$5$} (T3);
			\draw[-{StealthFill[fill=red]}, red] (L34) to[out=80,in=210] node[below, pos=0.5] {$5$} (T1);
			\end{tikzpicture}
			~~~~~~~~~~~~~~~~~~
		}
	}
	\caption{Dependency graph of deadlock cases.}
	\label{fig.dependency-graphs}
\end{figure}

\smartparagraph{Deadlocks solvable by segmentation.}
Coming back to the example given in Section~\ref{sec:overview}, we show its
dependency graph in Figure~\ref{fig.deadlock}. This deadlock is unsolvable by
the splitting technique. However, as discusses before, if we allow a packet to
be carried in the mix of the old and the new path of the same flow, this kind of
deadlock is solvable by using {\em segmentation}.

\ezsegway decomposes this deadlocked graph into two non-deadlocked 
dependency
graphs as shown in Figure \ref{fig.dependency-g1} and \ref{fig.dependency-g2},
hence enabling the network to be updated.

\begin{figure}[!t]
	\centering
	\subfloat[\scriptsize{$\NwSwitch_2$ with segmentation}]{
		\label{fig.dependency-g1}
		\scalebox{0.6}{
			~~~~
			\begin{tikzpicture}[scale=2, thick,
			StealthFill/.tip={Stealth[line width=1pt, scale=1.5]}, arrows={[round]}]
			
			\node [resource] (L23) at (0.5,0.6) {$\Link{2}{3}$:$0$};
			
			\node [move, draw=blue, very thick] (T1) at (2.5,0.6) {$\MovePath_{R.1}$};
			\node [move, draw=blue, very thick] (T3) at (2.5,0) {$\MovePath_{B.1}$};
			
			\draw[-{StealthFill[fill=black]}, black] (T1) -- node[above, pos=0.9] {$5$}  (L23);
			\draw[-{StealthFill[fill=red]}, red] (L23) -- node[above, pos=0.9] {$5$}  (T3);
			\end{tikzpicture}
			~~~~
		}
	}~
	\subfloat[\scriptsize{$\NwSwitch_3$ with segmentation}]{
		\label{fig.dependency-g2}
		\scalebox{0.6}{
			~~~~
			\begin{tikzpicture}[scale=2, thick,
			StealthFill/.tip={Stealth[line width=1pt, scale=1.5]}, arrows={[round]}]
			
			\node [resource] (L34) at (0.5,0) {$\Link{3}{4}$:$0$};
			
			\node [move, draw=blue, very thick] (T1) at (2.5,0.6) {$\MovePath_{R.2}$};
			\node [move, draw=blue, very thick] (T3) at (2.5,0) {$\MovePath_{B.2}$};
			
			\draw[-{StealthFill[fill=black]}, black] (T3) -- node[below, pos=0.9] {$5$}  (L34);
			\draw[-{StealthFill[fill=red]}, red] (L34) -- node[below, pos=0.9] {$5$}  (T1);
			\end{tikzpicture}
			~~~~
		}
	}
	\caption{Dependency graph of deadlock cases.}
	\label{fig.dependency-graphs-segment}
\end{figure}
\smartparagraph{Splittable deadlock.}
Assume that in our example we execute the update operation $\MovePath_G$ before
$\MovePath_R$. After $\Flow_G$ is moved to the new path, the dependency graph
becomes the one in Figure~\ref{fig.deadlock-split}. In this case, every link has
$3$ units of capacity but it is impossible to continue the update. However, if
we allow the traffic of flow $\Flow_R$ and $\Flow_B$ to be carried in both the
old path and the new path at the same time, we can move $3$ units of $\Flow_R$
and $\Flow_B$ to the new path and continue the update that enables updating the
remaining part of flows. 
The deadlock would not be splittable if the capacity of the relevant links was
zero, as shown in Figure~\ref{fig.deadlock}.

In the presence of a splittable deadlock, there exists a splittable flow
$\Flow_p$ and there is a switch $s$ in the new segment of $\Flow_p$. Switch $s$
detects the deadlock and determines the amount of $\Flow_p$'s volume that can be
split onto the new segment. This is taken as the minimum of the available
capacity on the $s$'s outgoing link and the necessary free capacity for the link
in the dependency cycle to enable another update operation at $s$.
An unsplittable deadlock corresponds to the state in which there is a cycle
in the dependency graph where each link has zero residual capacity and it
is not possible to release any capacity from the links in the cycle.

\smartparagraph{Congestion-free heuristic.}
Having categorized the space of possible deadlocks, we now introduce our
scheduling heuristic, called \schedulingheuristic, whose goal is to perform a congestion-free update as fast as possible.
The main goal is to avoid both unsplittable deadlocks, which can only be solved
by violating congestion-freedom, and splittable deadlocks, which require more
iterations to perform an update since flows are not moved in one single phase.

 \schedulingheuristic works as follows. It receives as input an instance of the network
update problem where each flow is already decomposed into segments. Each flow
segmentation is updated with \segmentedmigration, which means that each segment
is updated directly from its old path to the new one if there is enough
capacity. Hence, each segment corresponds to a network update node in the
dependency graph of the input instance. Each switch assigns to every segment
that contains the switch in its new path a priority level based on the following
key structure in the dependency graph. An update node $\MovePath$ in the
dependency graph is {\em critical} at a switch $s$ if $(i)$ $s$ is the first switch of
the segment to be updated, and $(ii)$ executing $\MovePath$ frees some capacity
that can directly be used to execute another update node operation that would
otherwise be not executable (\ie, even if every other update node operation
could be possibly executed). A {\em critical} cycle is a cycle that contains a 
critical update node.

\schedulingheuristic assigns low priority to all the segment (i.e., a network update
node) that do not belong to any cycle in the dependency graph.
These update operations consume useful resources that are needed in order to
avoid splittable deadlocks and, even worse, unsplittable deadlocks, which
correspond to the presence of cycles with zero residual capacities in the 
dependency graph, as previously described.
We assign medium priority to all the remaining segments
that belong only to non-critical cycles, while we assign higher priority to all
the updates that belong to at least one critical cycles. 
This guarantees that updates belonging to a critical cycle are executed
as soon as possible so that the risk of incurring in a splittable or
unsplittable deadlock is reduced. Each switch schedules its network operations
as follows. It only considers segments that needs to be routed through its
outgoing links. Among them, segment update operations with lower priority should
not be executed before operations with higher priority unless a switch detects
that there is enough bandwidth for executing a lower level update operations
without undermining the possibility of executing higher priority updates when
they will be executable. We run a simple Breadth-First-Search (\bfs) tree rooted
at each update operation node to determine which update operations belong to at
least one critical cycle.
In addition to the priority values of each flow, the updates must satisfy the
constraints imposed by \segmentedmigration, if there are any (i.e., there is
at least one {\sc InLoop} segment).

We can prove that \schedulingheuristic is \emph{correct} 
\if \TR 1
  (proof in Appendix~\ref{sec:correctness}),
\else
  (proof in the extended version of this paper~\cite{ez-segway-tr}),
\fi
i.e., 
as long as it there is an executable network update operation there is no congestion
in the network, and
that the worst case complexity for identifying a critical cycle for a
 given update operation is $O(|\MovePathSet|{+}|\LinkSetDependent|{+}|\MovePathSet|{\times}|\LinkSetDependent|){\simeq}O(|\MovePathSet|{\times}|\LinkSetDependent|)$. Consequently, for all update
 operations the complexity is $O(|\MovePathSet|^2{\times}|\LinkSetDependent|)$.

\begin{theorem}\label{theo:deadlock-free}
 \schedulingheuristic is correct for the network update problem.
\end{theorem}

\subsection{Distributed coordination }
\label{sec:coordination}

We now describe the mechanics of coordination during network updates.

\smartparagraph{First phase: the centralized computation.}
 As described in Sect.~\ref{sec:reliability-segmentation}, to avoid black-holes
 and forwarding loops, each segment
can be updated using a \basicmigration technique unless there are {\sc InLoop} 
segments.
These dependencies are computed by the controller in the initialization phase and
transmitted to each switch in order to coordinate the network update correctly.

As described in Sect.~\ref{sec:ez-segway-scheduling}, the scheduling
mechanism assigns one out of three priority levels to each segment update
operation. The centralized controller is responsible for performing this
computation and sending to each switch a \MsgInstallUpdate message that
encodes the segment update operations (and their dependencies) that must be
handled by the switch itself in addition to the priority level of the segment 
update operations.

\smartparagraph{Second phase: the distributed computation.}
 Each switch $s$ receives from the centralized controller the following information
regarding each single segment $S$ that traverses $s$ in the new path:
its identifier, its priority level
(i.e., high, medium, or low), its amount of traffic volume, the identifier of the switch that precedes (succeeds) it along the new and old path, 
and whether the receiving switch is the initial or
final switch in the new path of $S$ and in the old and new path of the 
flow that contains $S$ as a segment. If the segment is
of type {\sc InLoop}, the identifier of the segment that has to be updated 
before this one is also provided.  Each switch in addition knows the capacity of its outgoing
links and maintains memory of the amount of capacity that is used by the flows
at each moment in time.

Upon receiving this information from the controller, each switch
performs the following initial actions for each segment $S$ that
it received: it installs the new path and it removes the old path. The messages exchanged by the switches for performing
these three operations  are described in detail the next three paragraphs. We
want to stress the fact that all the functionalities executed in the switches
consist of simple message exchanges and basic ranking of update operations that
are computationally inexpensive and easy to implement in currently available
programmable switches. Those operations are similar to those performed by MPLS 
to install labels in the switches~\cite{rfc-3209}. Yet, MPLS does not provide any 
mechanism to schedule the path installation in a congestion-free manner.

\smartparagraph{Installing the new path of a segment. }
The installation of the new path is performed by iteratively reserving along the reversed
new path the bandwidth required by the new flow.
The last switch on the new path of segment $S$ sends a \MsgGoodToMove message to 
his predecessor along the new path. 
The goal of this message is to acknowledge the receiver that the downstream
path is set up to receive the traffic volume of $S$.
Upon receiving a \MsgGoodToMove message for a segment $S$, a switch
checks if there is enough bandwidth on the outgoing link to execute the 
update operation. If not, it waits that enough bandwidth will be available when
some flows will be removed from the outgoing link. In that case, it checks if there are segment update operations
that require the outgoing link and have higher priority than $S$. If 
not, the switch executes the update operation. Otherwise, it checks whether the residual capacity
of the link minus the traffic volume of the segment is enough to
execute in the future all the higher priority update operations. In that case, the switch 
executes the update operations.
 If the switch successfully performs the update operation, it updates the residual capacity 
of the outgoing link and it sends a \MsgGoodToMove message to its predecessor
along the new path of $S$. If the switch has no predecessor along the new path of
$S$, i.e., it is the first switch of the new path of $S$, it sends a \MsgRemoving message
to its successor in the old path.
 If the receiving switch is the last switch of an {\sc InLoop} segment $S'$, it sends a 
\MsgGoodToMove message to its predecessor on $dep(S')$.

\smartparagraph{Removing the old path of a segment. }
Upon receiving a \MsgRemoving message for a segment $S$, if the receiving switch
 is not a switch in common with the new path that has not yet installed the new flow entry,
 it removes the entry for the old path and it forwards the message to its successor in the old path.
 Otherwise, it puts on hold the \MsgRemoving message until it installs the new path.
 If the switch removed the old path, it updates the capacity of its outgoing links and checks whether there 
 was a dependency between 
 the segment that was removed and any segment that can be executed at the 
 receiving switch. In that case, it executes these update operations according
 to their priorities and the residual capacity (as explained above) and propagates
 the \MsgGoodToMove that were put on hold.

\subsection{Dealing with failures}

While the network update process must deal with link and switch failures, we
argue that doing so from within the switches simply introduces unwarranted
complexity. Thus, we deliberately avoid dealing with failures in our distributed
coordination update mechanism.

As with a centralized approach, if a switch or link fails during an update, a
new target configuration must be computed. Recall that the \ezsegway is
responsible for updating from an initial configuration to the final one but not
for computing them.
We believe that controller is the best place to re-compute a new global desired
state and start a new update. Note that in the presence of a switch or link
failure, our update process stops at some intermediate state. Once the
controller is notified of the failure,it halts the updates and queries the switches to 
know which update operations were performed and uses this information to 
reconstruct the current network state and compute the new desired one.

This process can  be optimized to minimize recovery delay. A natural
optimization is to have switches asynchronously sending a notification to the
controller upon performing an update operation, thus enabling the controller to
keep closer in sync with the data plane state.

As for the messages between switches, we posit that these packets are sent with
the highest priority so that they are not dropped due to congestion and that
their transmission is retried if a timeout expires before an acknowledgment is
received. When a message is not delivered for a maximum number of times, we
behave as though the link has failed.
\section{Implementation}
\label{sec:implementation}

We implemented an unoptimized \ezsegway prototype written as 6.5K LoC in Python.
The prototype consists of a \emph{global controller} that runs centrally as a
single instance, and a \emph{local controller} that is instantiated on every
switch. The local controller is built on top of the Ryu controller~\cite{ryu}. The 
local controller, which executes our
\schedulingheuristic algorithm, connects and manipulates the data-plane state
via OpenFlow, and sends/receives UDP messages to/from other switches. Moreover,
the local controller communicates with the global controller to receive the
pre-computed scheduling information messages and to send back an acknowledgment
once an update operation completes.

\section{Prototype Evaluation}
\label{sec:evaluation}

We now evaluate the update time performance of our prototype through
emulation in Mininet~\cite{mininet-hifi}.
We compare \ezsegway with a \textit{Centralized} approach, inspired by 
\dionysus~\cite{Jin+DSN+2014}, that runs the same
scheduling algorithm of \ezsegway but coordination among the switches is
delegated to a central controller. To make the comparison fair, we 
empower Centralized with segmentation and splitting volume support, both of
which do not exist in \dionysus.

\smartparagraph{Experimental setup.} We run all our experiments on a dedicated server with 16 cores at
2.60~GHz with hyper-threading, 128~GB of RAM and Ubuntu Linux 14.04. The computation in the central controller is parallelized across all cores. Deadlocks are detected by means of a timeout, which we set to 150 ms.

 We consider two real WAN topologies: $B4$~\cite{b4} -- the
 globally-deployed software defined WAN at Google --  
 and the layer 3 
 topology of the $Internet2$ network~\cite{Internet2}. 
Without loss of generality, we assume
link capacities of 1~Gbps. We place the controller at the centroid switch of
the network, \ie, the switch that minimizes the maximum latency towards the
farthest switch. We compute the latency of a link based on the
geographical distance between its endpoints and the 
signal propagation speed through optical fibers (\ie, $\sim$200,000~km/s).

\begin{figure}[!t]
	\centering
	\includegraphics[width=0.48\textwidth]{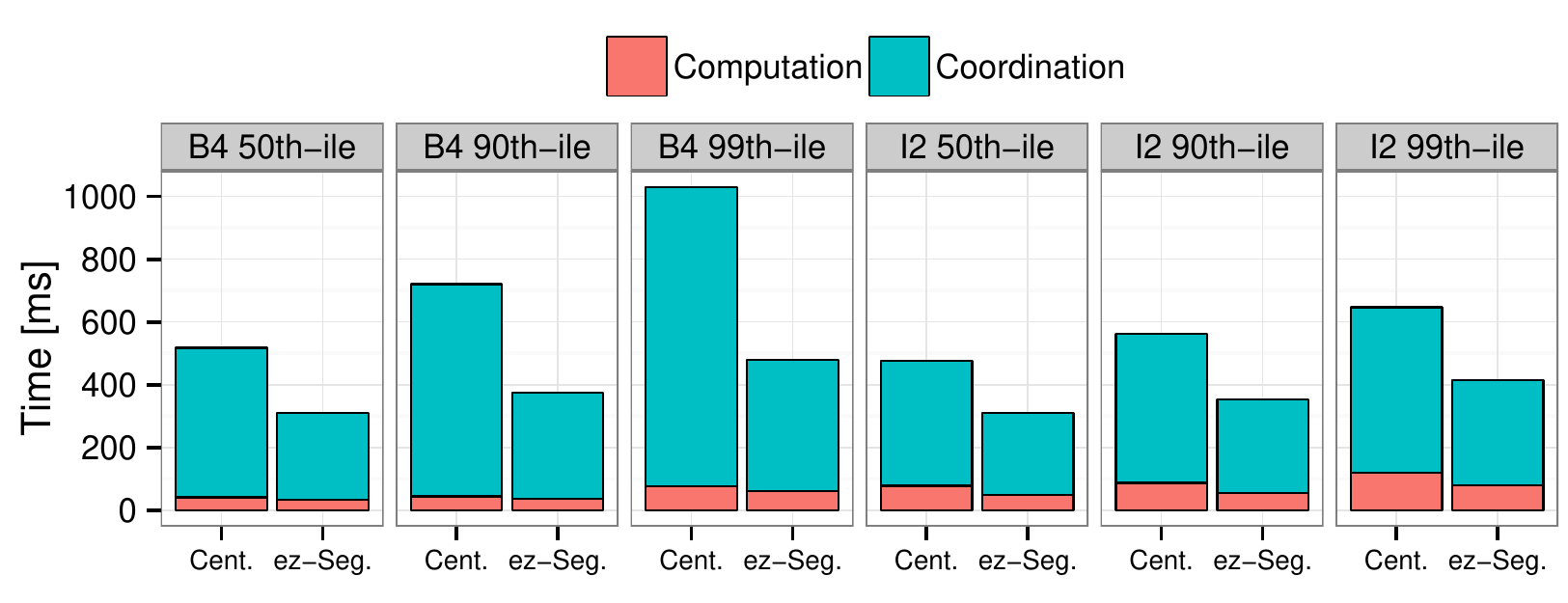}
	\caption{Update time of \ezsegway~versus a Centralized approach for 
		$B4$ and $Internet2$.}
	\label{fig:update_time_exp}
\end{figure}

Similarly to the methodology in~\cite{He.CoNext07}, we generate all flows of a 
network
configuration by selecting non-adjacent source and destination switch pairs at
random and assigning them a traffic volume generated according to the gravity
model~\cite{Roughan.SIGCOMM05}. For a given source-destination pair 
$(s,d)$, we compute $3$ paths from $s$ to $d$ and equally spread the 
$(s,d)$ traffic volume among these $3$ paths. To compute each path, we first 
select a third transit node $t$ at random and then 
we compute a
cycle-free shortest path from $s$ to $d$ that traverses $t$. If it does not exist, 
we
choose a different random transit node. We guarantee that the latency of the
chosen paths is at most a factor of $1.5$ greater than the latency of the
source-destination shortest path. Starting from an empty 
network configuration that
does not have any flow, we generate $1,000$ continuous network 
configurations. If the volumes of the flows in a configuration exceed the 
capacity of at least one link, we iteratively remove flows until the 
remaining ones fit within the given link capacities. 
The resulting network updates consist of a modification of at least 
(at most) $176$ and $188$ ($200$ and $276$) flows in B4 and 
Internet2, respectively. 
The experiment goes through these $1,000$ sequential update configurations 
and we measure both the time for completing each network update and each 
individual flow update. The
network update completion time corresponds to the slowest flow completion time
among the flows that were scheduled in the update.

We break down the update time by the amount of computation performed at the
central controller for the scheduling \textit{computation} of the network updates and
due to \textit{coordination} time among the switches to perform the network update.
Our coordination time is measured as the interval between the time when the
first controller-to-switch is sent until when the last switch notifies the controller
that it has completed all update operations.
We report the 50$^{th}$, 90$^{th}$, and 99$^{th}$ percentile update time across all the update configurations.

\smartparagraph{\ezsegway outperforms Centralized.}
Figure~\ref{fig:update_time_exp} shows the results for the update 
completion time. For B4, \ezsegway is $45\%$ faster than Centralized
in the median case and $57\%$ faster for the $99^{th}$ percentile of updates, 
a crucial 
performance gain for the scenarios considered.  We observed a 
similar trend for I2, where the completion time is  $38\%$ shorter than 
Centralized at the $90^{th}$ percentile.
It should be noted that we have adopted almost worst conditions for
\ezsegway in this comparison in that the central controller is ideally placed at the network centroid.
Importantly, our approach is able to complete all the network updates, despite
overcoming $203$ and $68$ splittable deadlocks in B4 and Internet2,
respectively. Also, \ezsegway exchanges 35\% of the messages sent by
Centralized.

\smartparagraph{Controller-to-switch communication is the 
Centralized bottleneck.} We note that the computation time at the controller 
represents a small
fraction of the overall update time (\ie, $\le20\%$), showing that the main
bottleneck of a network update is the coordination time among 
the network devices, where \ezsegway significantly outperforms 
Centralized. As we show later with micro-benchmarks on a real switch in
\S\ref{sec:micro-benchmark}, the computation within the switch is in the order
of $3{-}4$ ms. This means that our approach is essentially limited by the
propagation latency of the speed of light, which cannot be improved with better scheduling algorithms.

\begin{figure}[!t]
\centering
		\includegraphics[width=0.48\textwidth]{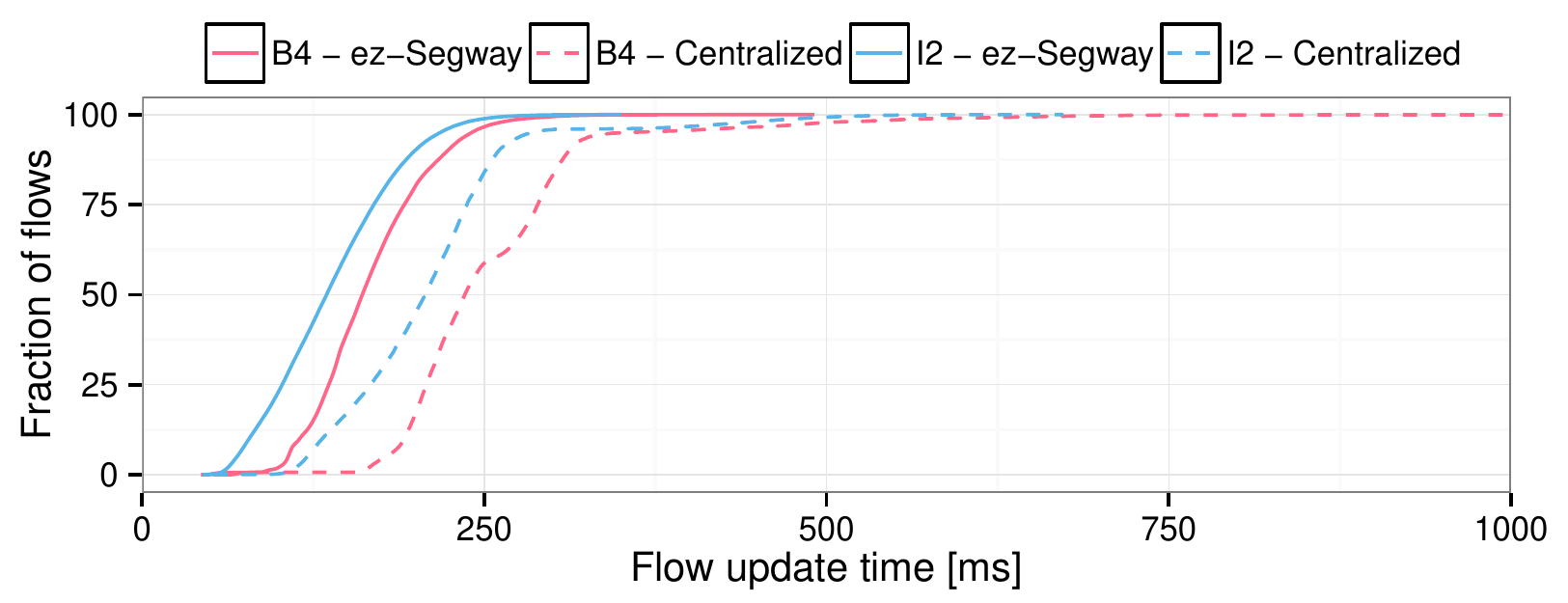}
	\caption{Flow update time of \ezsegway~versus a Centralized approach 
	for $B4$ and $Internet2$.}
	\label{fig:flow_update_time}
\end{figure}

\begin{figure*}[t!]
	\begin{center}
		\includegraphics[width=0.95\textwidth]{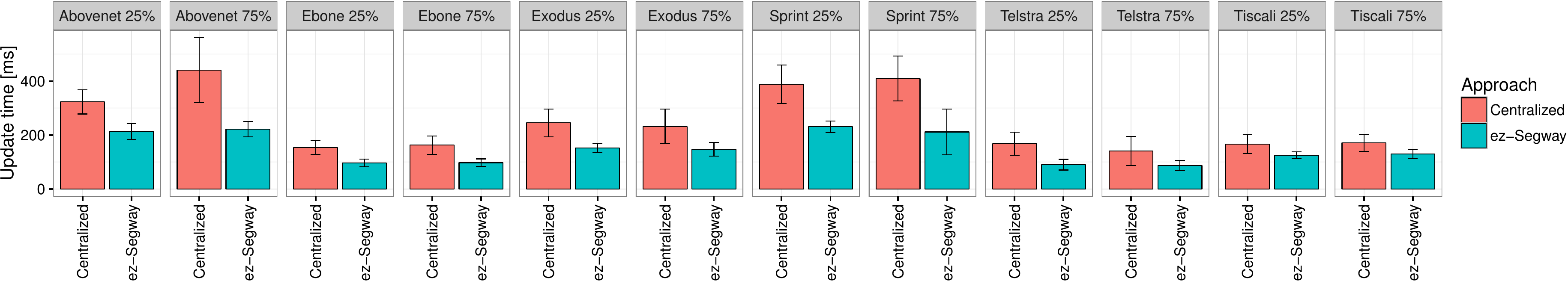}
	\end{center}
	\caption{Update time of \ezsegway versus Centralized for various simulation settings.}
	\label{fig:update_time}
\end{figure*}

\smartparagraph{\ezsegway speeds up the slowest flows.}
Figure~\ref{fig:flow_update_time} shows the distribution of the update completion time of
individual flows across all network updates. We observe that \ezsegway (solid lines) 
not only consistently outperforms the
centralized approach (dashed lines) in both B4 and Internet2, but it also reduces the long-tail visible in Centralized, 
a fundamental improvement to network's responsiveness and recovery after failures.
We observed that the maximum flow update time is $456$~ms in
\ezsegway and $1634$~ms in Centralized for B4
and $350$~ms in \ezsegway and
$673$~ms in Centralized for Internet2. The mean flow update time is $26\%$ 
slower and $51\%$ slower with Centralized for both B4 and Internet2, 
respectively.

\begin{table}[t]
	\vspace{.5em}
	\centering
	\begin{tabular}{|l||c|c|}
		\hline			
		Overhead & Mean & Max. \\
		\hline
		Rule & $1.37(\pm 0.57)$ & $4$ \\
		\hline
		Message & $3.16(\pm3.62)$ & $41$ \\
		\hline  
	\end{tabular}
	\caption{Overheads for flow splitting operations.}\label{tbl:overheads} 
\end{table}

\smartparagraph{\ezsegway overcomes deadlocks with low overheads.}
We now measure the cost of resolving a deadlock by splitting flows in terms of
forwarding table overhead and message exchange overhead in the same experimental
setting used so far. Splitting flows increases the number of forwarding entries
since both the new and the old flow entries must be active at the same time in
order to split a flow among the two of them. Moreover, splitting flows requires
coordination among switches, which leads to an overhead in the message
communication. We compare \ezsegway against a version of \ezsegway that does not
split flows and completes an update by congesting links. 
Table~\ref{tbl:overheads} shows that the average (maximum) number
of additional forwarding entries that are installed in the switches is $1.37$ 
($4$). As for the message exchange overhead, we observe that the average  
(maximum) number of
additional messages in the entire network is $3.16$ ($41$).
Thus, we conclude that the overheads are negligible.

\section{Large-scale simulations}
\label{sec:simulation}

We turn to simulations to study the performance of our approach on large 
scale Internet topologies.
We compare \ezsegway against the Centralized \dionysus-like approach, which we 
defined in \S\ref{sec:evaluation}.
We measure the total update time to install updates on $6$ real topologies
annotated with link delays and weights as inferred by the RocketFuel
engine~\cite{Mahajan.IMW02}. We set link capacities between $1{\sim}100$ Gbps,
inversely proportional to weights. We place the controller at the centroid
switch.

We generate flows using the same approach as in our prototype experiments.
We run simulations for a number of flows varying from 500 to 1,500 and
report results for 1,000 flows as we observed \emph{qualitatively} similar
behaviors. Since the \emph{absolute} values of 
network update time strongly depend 
on the number of flows, we focus on the performance gain factor.
We generate updates 
by simulating link failures that
cause a certain percentage $p$ of flows to be rerouted along new shortest
paths. We ran experiments for $10\%$, $25\%$, $50\%$, and $75\%$; for brevity,
we report results for $25\%$ and $75\%$. For every setting of topology, flows,
and failure rate, we generate $10$ different pairs of old and new network
configurations, and report the average update completion time and its 
standard deviation. 
Figure~\ref{fig:update_time} shows our results, which demonstrate that
\ezsegway reduces the update completion time by a factor of $1.5-2$. 
In practice, the ultimate gains of using \ezsegway will depend on specific 
deployment scenarios (e.g., WAN, datacenter) and might not be significant in 
absolute terms in some cases.

\section{Hardware Feasibility}
\label{sec:micro-benchmark}

\smartparagraph{OpenFlow.}
To assess the feasibility of deploying \ezsegway on real programmable switches,
we tested our prototype on a Centec $V580$ switch~\cite{centec}, which runs a low-power
$528$MHz core Power PC with $1$GB of RAM, wherein we execute the \ezsegway
local-controller. Our code runs in Python and is not
optimized. We observe that executing the local controller in the switch required
minimal configuration effort as the code remained unchanged.

We first run a micro-benchmark to test the performance of the switch for running
the scheduling computation. Based on Internet2 topology tests, we
extract a sequence of $79$ updates with $50$ flows each, that we send to the
switch while we measure compute time. We measure a mean update processing time
of $3.4$ ms with $0.5$ ms standard deviation. This indicates that our approach
obtains satisfactory performance despite the low-end switch CPU. Moreover, the
observed CPU utilization for handling messages is negligible.

As a further thought for optimizing \ezsegway, we asked ourselves whether
current switches are suitable for moving the scheduling computation from the
centralized controller to the switches themselves. This would be a natural and
intuitive optimization since our scheduling algorithm consists of an iterative
computation of the update operations priorities for each switch. This would be
beneficial to the system when the number of switches in the network is much
larger than the number of cores in the central controller. We found that
processing the dependency graph for the same sequence of 79 updates of the
Internet2 topology with $\sim$200 flows in the entire network with our
heuristic takes on avarage 22 ms with 6.25 ms standard deviation.

\smartparagraph{P4.}
To further illustrate the feasibility of realizing simple network update
functionality in switches, we explored implementing \ezsegway in
P4~\cite{p4}. We regard this as a thought experiment that finds a 
positive answer. However, the need to delegate update functions to the switch
data plane is unclear since the main performance bottleneck in our results is
due to coordination.

Using P4, we are able to express basic \ezsegway functionality as a set of
handlers for \MsgInstallUpdate, \MsgGoodToMove, and \MsgRemoving messages. The
\MsgInstallUpdate handler receives a new flow update message and updates the
switch state. If the switch is the last switch on the new path, it performs
the flow update and sends a \MsgGoodToMove message to its predecessor. Upon
receiving a \MsgGoodToMove message, a switch checks whether the corresponding
update operation is allowed by ensuring that the available capacity on the new
path is higher than the flow volume. In that case, the switch performs the
update operation, updates the available capacity, and forwards the
\MsgGoodToMove message to its predecessor (until the first switch is reached).
When the lack of available capacity prevents an update operation, the switch
simply updates some bookkeeping state to be used later. Once a
\MsgGoodToMove message is received, the switch also attempts to perform any
update operation which was pending due to unavailable capacity. Upon receiving
a \MsgRemoving message, the switch simply uninstalls the flow state and
propagates the message to the downstream switch till the last one. For sake of
presentation, the above does not describe priority handling.

While the above functionality is relatively simple, some of the language
characteristics of P4 make it non-trivial to express it. P4 does not have
iteration, recursion, nor queues data structures. Moreover, P4 switches
cannot craft new messages; they can only modify fields in the header
of the packet that they are currently processing~\cite{Dang.CCR16}. To
overcome these limitations of the language, we perform loop unrolling of the
\ezsegway logic and adopt a header format that contains the union of all
fields in all \ezsegway messages.
\if \TR 1
 We report the detailed P4 instructions needed to implement \ezsegway in Appendix~\ref{sec:p4-implementation}.
\else
 More details on our P4 implementation appear in the extended version of this paper~\cite{ez-segway-tr}.
\fi

\section{Related work}

The network update problem has been widely studied in the recent
literature~\cite{Jin+DSN+2014, Reitblatt+ANU+2012, Liu+ZUD+2013,
McClurg.PLDI15, Peresini.HOTSDN14, katta2013incremental, Canini.INFOCOM15,
cupid-infocom-2016, Brandt2016On, flip-infocom-16, timeflip, time4-infocom-16}. These
works use a centralized architecture in which an SDN
controller computes the sequence of update operations and actively coordinates
the switches to perform these operations. The work in~\cite{time4-infocom-16}
relies on clock-synchronization for performing fast updates, which are, however,
non-consistent due to synchronization imprecisions and non-deterministic delays
in installing forwarding rules~\cite{Jin+DSN+2014, kuzniar-pam-15, DBLP:conf/sosr/HeKGDPALT15}.
In contrast, \ezsegway speeds up network updates by delegating the task of
coordinating the network update to the switches. To the best of our knowledge,
\ezsegway is the first work that tackles the network update problem in a
decentralized manner.

In a sense, our approach follows the line of recent proposals to centrally
control distributed networks like Fibbing~\cite{fibbing-2015}, where the
benefits of centralized route computation are combined with the reactiveness of 
distributed approaches. We discuss below the most relevant work with
respect to our decentralized mechanism.

\dionysus~\cite{Jin+DSN+2014} is centralized scheduling algorithm that updates
flows atomically (\ie, no traffic splitting). \dionysus computes a graph to
encode dependencies among update operations. This dependency graph is used by
the controller to perform update operations based on dynamic conditions of the
switches. While \ezsegway also relies on a similar dependency graph, the
controller is only responsible for constructing the graph whereas the switches
use it to schedule update operations in a decentralized fashion.

Recent work~\cite{cupid-infocom-2016} independently introduced flow segmentation
and traffic splitting techniques similar to ours~\cite{Nguyen.ANRW16}. However,
their work and problem formulation focuses on minimizing the risk of incurring
in a deadlock in the centralized setting. In contrast, \ezsegway develops flow
segmentation to increase the efficiency of network update completion time.

From an algorithmic perspective, \cite{Jin+DSN+2014} and~\cite{Brandt2016On}
showed that, without the ability to split a flow, several formulations of the
network update problem are NP-hard. The work in~\cite{Brandt2016On} is the only one that
provides a poly-time algorithm that is correct and complete for the
network update problem. However, their model allows flows to be moved on any
possible intermediate path (and not only the initial and final one). Moreover,
there are several limitations. First, the complexity of the algorithm is too
high for practical usage. Consequently, the authors do not evaluate their
algorithm. Last, this work also assumes a centralized
controller that schedules the updates.

\section{Conclusion}
This paper explored delegating the responsibility of executing consistent
updates to the switches. We proposed \ezsegway, a mechanism that achieves faster
completion time of network updates by moving simple, yet clever coordination
operations in the switches, while the more expensive computations are performed
in the centralized controller. Our approach enables switches to engage as active
actors in realizing updates that provably satisfy three properties: black-hole
freedom, loop freedom, and congestion freedom.

In practice, our approach leads to improved update times, which we quantified
via emulation and simulation on a range of network topologies and traffic
patterns. Our results show that \ezsegway improves network update times
by up to $45\%$ and $57\%$ at the median and the $99^{th}$ percentile, 
respectively.  We also
deployed our approach on a real OpenFlow switch to demonstrate the feasibility
and low computational overhead, and implemented the switch functionality in P4.

\vspace{1pt}
{\footnotesize
\smartparagraph{Acknowledgements.}
\if \TR 0
We would like to thank the anonymous reviewers and our shepherd Jia Wang for their comments.
\fi
We are thankful to Xiao Chen, Paolo Costa, Huynh Tu Dang, Petr Kuznetsov, Ratul Mahajan,
Jennifer Rexford, Robert Soul\'{e}, and Stefano Vissicchio for their helpful
feedback on earlier drafts of this paper. This research is (in part) supported
by European Union's Horizon 2020 research and innovation programme under the
ENDEAVOUR project (grant agreement 644960).}

\bibliographystyle{abbrv} 
\if \TR 1
\begin{footnotesize}

\end{footnotesize}
\else

\fi

\if \TR 1
\appendix
{
\newcommand{\rephrase}[3]{\noindent\textbf{#1 #2}.~\emph{#3}}

\section{Correctness}
\label{sec:correctness}

We first remind the three properties that \ezsegway aim to preserve during a network update:

\begin{itemize}
	\item {\em black-hole freedom}: For any flow, no packet is unintendedly dropped in the network.
	\item {\em loop-freedom}: No packet should loop in the network.
	\item {\em congestion-freedom}: No link has to carry a traffic greater than its capacity.
\end{itemize}

\spacebeforeparagraph

\smartparagraph{Notation and terminology. }
We first introduce some useful notation and terminology. Let $u_0,\dots,u_k$ be the sequence of $k$ update operations executed during a network update and et $t_0,\dots,t_k$ be the
sequence of time instants where these operations are performed, with $t_i\le t_j$ for any $i<j$. 
Since each switch performs a single update operations at a time, 
w.l.o.g., we assume that at each time instant only one operation is
performed. We assume that, at time $t_0$, the controller sends the precomputed information to all the switches.
	Let $F$ be a flow from a switch $s_i$ to $s_f$ that is being updated
	from an old path $P_{old}=(s_o^1\dots s_o^n)$ to a new path $P_{new}=(s_n^1\dots s_n^m)$. 
	For each switch $s \in \NwSwitchSet$, let $next(F,s,t_j)$ be
	the next-hop of $F$ according to the rules installed at $s$ at time $t_j$, where $next(F,s,t_j)=\epsilon$ if there is
	no rule installed for that specific flow. Let $path\_forwarding(F,s,t_j)$ be the (possibly infinite) sequence of switches that is
	traversed by a packet belonging to $F$ that is received by switch $s$ at time $t_j$. Observe that the incoming port of a packet is not relevant in \ezsegway.

Let $prec_{old}(s)$ ($prec_{new}(s)$)  be the predecessor vertex of $s$
in $P_{old}$ ($P_{new}$) and $succ_{old}(s)$ ($succ_{new}(s)$)  be the successor vertex of $s$
in $P_{old}$ ($P_{new}$).
Let $prec_{old}^*(s)$ ($prec_{new}^*(s)$) be the set of switches between  $s_i$
and $s$ in $P_{old}$ ($P_{new}$) and   let $succ_{old}^*(s)$ ($succ_{new}^*(s)$) be the 
set of switches between $s_f$ and $s$.

\subsection{Basic-Migration Correctness}

We first prove that the \basicmigration operation is guaranteed to complete in the absence of congestion. Let $F$ be any flow that has to be updated from $P_{old}$ to $P_{new}$.

\begin{lemma}\label{lemm:basic-migration-1}
	 \basicmigration guarantees that every switch on the new path eventually receives a \MsgGoodToMove message.
\end{lemma}

\begin{proof}
	Suppose, by contradiction, that the statement is not true, i.e., there exists a switch on the 
	new path that does not receive a \MsgGoodToMove message. Let $s_n^h$ be the closest switch to $s_f$ on the new path that does not receive a \MsgGoodToMove message. This means that 	$s_n^{h+1}$ received a \MsgGoodToMove message but it did not send it to $s_n^h$. This is a contradiction since we assumed that it is not possible to congest a link during the update. 
	Hence, the statement of the lemma is true.
\end{proof}

\begin{lemma}\label{lemm:basic-migration-removing-1}
	\basicmigration guarantees that every switch on the old path eventually receives a \MsgRemoving message.
\end{lemma}

\begin{proof}
	By Lemma~\ref{lemm:basic-migration-1}, $s_i$ is guaranteed to receive a \MsgGoodToMove message.
	When $s_1$ receives such messages, it switches its forwarding path to the new one, it removes the old path entry, and it sends a \MsgRemoving message
	to its successor of the old path. 
	Suppose, by contradiction, that the statement is not true, i.e., there exists a switch of the 
	old path that does not receive a \MsgRemoving message. Let $s_o^h$ be the closest switch to 
	$s_i$ on the old path that does not receive a \MsgRemoving message. 
	This means that 	$s_n^{h-1}$ received a \MsgRemoving message but it did not send it to $s_n^h$, which is a contradiction.
	Hence, the statement of the lemma is true.
\end{proof}

\begin{lemma}\label{lemm:basic-migration-2}
	 \basicmigration guarantees that any packet routed from a vertex that received a \MsgGoodToMove reaches $s_f$ in a finite number of steps.
\end{lemma}

\begin{proof}
		We prove by induction on $t_0,\dots,t_k$ that for each switch $s_n^h$ of the new path that received a \MsgGoodToMove message at time $t_j > t_0$ we have that $next(s_n^h,t_k)=s_n^{h+1}$ holds, with $t_k\ge t_j$. That is, $path\_forwarding(s,t_k)=(s_n^h\dots s_n^m)$.
		At time $t_0$, the statement trivially holds since none of the switches received a \MsgGoodToMove.
		At time $t_j>t_0$, if $s_n^h$ received a \MsgGoodToMove message at any time 
		before $t_j$, then by induction hypothesis $next(s_n^h,t_j)=s_n^{h+1}$ holds. 
		If  $s_n^h$ receives a \MsgGoodToMove at $t_j$, it installs the new forwarding rules $next(s_n^h,t_j)=s_n^{h+1}$.
		Hence, the statement of the lemma holds.
\end{proof}

\begin{lemma}\label{lemm:basic-migration-no-black-holes}
	\basicmigration is black-hole-free.
\end{lemma}

\begin{proof}
	A black-hole is created if a switch removes the old path while there are some packets that
	still needs to be routed through that old path. Since the \MsgRemoving is sent once $s_i$ 
	switches to the new path, any packet sent after that message is create is forwarded on the new path. Hence, this guarantees that a packet is never dropped.
\end{proof}

\vspace{2mm}
\rephrase{Theorem}{\ref{theo:basic-migration}}{
	\basicMigrationCorrectnessStatement
}
\vspace{2mm}

\begin{proof}
	We first show that \basicmigration prevents forwarding loops.
	Consider a packet that is routed along the old path. Two cases are possible: it either reaches $s_f$ along it or it reaches a switch that installed the new path forwarding entry. In the latter case, by Lemma~\ref{lemm:basic-migration-2},
	it is guaranteed to reach $s_f$. 
	
	By Lemma~\ref{lemm:basic-migration-no-black-holes}, we have that \basicmigration guarantees 
	that black-holes are preventing.
	
	We now prove that the update operation terminates, we first need to show that $s_i$ eventually switches to the new path. 
	By Lemma~\ref{lemm:basic-migration-1}, $s_i$ is guaranteed to receive the \MsgGoodToMove message
	and therefore it will switch to the new path. We then need to prove that the old path is eventually
	removed from the network. By Lemma~\ref{lemm:basic-migration-removing-1}, each switch of 	the old path is guaranteed to receive the \MsgRemoving message and, in turn, to remove the old path	entry.
	
	Hence, the statement of the theorem is true.
\end{proof}

\subsection{Segmented-Migration Correctness}

\spacebeforeparagraph
\smartparagraph{Heuristic reminder. }
\segmentationHeuristicDescritption

 For each pair $(r,s) \in \NwPairSwitchSet_R$, let $\NwSwitchSet_{(r,s)}$ be the set
 of switches contained in the subpath from $r$ to $s$ of $P_{old}$, endpoints included.
We order the pairs in $\NwPairSwitchSet_R$ as follows. We have that $(r,s) < (r',s')$
 if $r$ is closer to $s_i$ than $r'$ on the old path.
 Recall that the pairs in $\NwPairSwitchSet_R$ are chosen in such a way that the 
 common switches that connects each pair of switches do not overlap each other.

\begin{lemma}\label{lemm:segmented-migration-no-loops-between-pairs}
 For any possible execution of \segmentedmigration, for any time $t_j>t_0$, for
 any pair $(r,s) \in \NwPairSwitchSet_R$, for any switch $u \in \NwSwitchSet_{(r,s)}$,
 $path\_forwarding(r,t_j)$ does not traverse any switch in $(prec_{old}^*(r) \cup prec_{old}^*(r)) \setminus \NwSwitchSet_{(r,s)}$.
\end{lemma}
  
\begin{proof}
	Suppose, by contradiction, that the statement of the lemma is not true, \ie, there exists
	an execution of \segmentedmigration such that at a certain time instant $t_j>t_o$, there
	exists a pair $(r,s) \in \NwPairSwitchSet_R$ such that for a switch $u \in \NwSwitchSet_{(r,s)}$ we have that
	$path\_forwarding(r,t_j)$ traverses a switch $w$ in $(prec_{old}^*(r) \cup prec_{old}^*(r))  \setminus \NwSwitchSet_{(r,s)}$.
	Let $w$ be the closer vertex to $u$ on $path\_forwarding(r,t_j)$.
	Let $(v_1\dots,v_n)$ be the subpath of $path\_forwarding(u,t_j)$ from $u$ to $w$. 
	Since $w$ is a closer to $s_i$ than $u$ and routing along the old path cannot get
	closer to $s_i$, $path\_forwarding(r,t_j)$ must contain a subpath $P^*$ that connects a switch $x$ in
	$\NwSwitchSet_{(r,s)}$ to $w$. This means that $x$ and $w$ appear in reverse order in
	the new and old path. Let $(r',s') $ be a pair such that $w \in \NwSwitchSet_{(r',s')}$. 
	We can then create a new pair $(r',s)$ that spans at least both $\NwSwitchSet_{(r,s)}$
	and $\NwSwitchSet_{(r',s')}$, which is a contradiction since, by Property (iii) of the set 
	of pairs $\NwPairSwitchSet_R$, the cardinality of $\NwPairSwitchSet_R$ is minimal.
	Hence, the statment of the lemma is true.
\end{proof} 
  
Lemma~\ref{lemm:segmented-migration-no-loops-between-pairs} guarantees that if
a forwarding loop exists, it must be confined within the two segments created by
the switches of a single pair in $\NwPairSwitchSet_R$. We now prove that the dependencies
among these segments prevent any forwarding loop.
 
 \begin{lemma}\label{lemm:segmented-migration-forwarding-loop-freedom}
For any possible execution of \segmentedmigration, for any time $t_j>t_0$, for
any pair $(r,s) \in \NwPairSwitchSet_R$, for any switch $u \in \NwSwitchSet_{(r,s)}$,
$path\_forwarding(r,t_j)$ is finite.
 \end{lemma}

\begin{proof}
	By Lemma~\ref{lemm:segmented-migration-no-loops-between-pairs},
	a forwarding loop can only traverse switches that belong to the old
	path of the {\sc NotInLoop} segment $S_r$ of a switch $r$ and the old
	path of the {\sc InLoop} segment $S_s$ of a switch $s$ for a pair $(r,s) \in
	\NwPairSwitchSet_R$. However, $s$ only switches to its new path
	when it receives the \MsgRemoving message of $S_r$ because
	$dep(S_s)=S_r$. Since Theorem~\ref{theo:basic-migration} guarantees that 
	\basicmigration is forwarding-loop-free, a forwarding loop cannot happen and the
	statement of the theorem is true.
\end{proof}
 
\begin{lemma}\label{lemm:segmented-migration-black-hole-freedom}
\segmentedmigration is black-hole free.
\end{lemma}	

\begin{proof}
 Since \segmentedmigration is an execution of independent \basicmigration, by Lemma~\ref{lemm:basic-migration-no-black-holes}, we have that a packet is never dropped during an update. We only need to check that the last
 switches of each segment do not drop a packet. By construction of the segments, each of
 these switches is a common switch between the old and the new path, thus they always
 have either the old or new forwarding entry installed, which proves the statement of the theorem.
\end{proof}

\vspace{2mm}
\rephrase{Theorem}{\ref{theo:segmentation-anomaly-freeness}}{
	\segmentedMigrationCorrectnessStatement
}
\vspace{2mm}

\begin{proof}
	It easily follows from Lemma~\ref{lemm:segmented-migration-forwarding-loop-freedom} 
	and~\ref{lemm:segmented-migration-black-hole-freedom}.
\end{proof}

\subsection{Congestion-freedom Correctness}

\begin{lemma}\label{lemm:ez-schedule-correctness}
\schedulingheuristic is correct for the network update problem.
\end{lemma}
	
\begin{proof}
	Since each switch moves to the new path only if its outgoing capacity is 
	at least the volume of the flow that has to be routed through the new path,
	it immediately follow that the \schedulingheuristic heursistic is correct.
\end{proof}

\section{Implementation in P4}
\label{sec:p4-implementation}
\lstset{xleftmargin=10pt}
\lstset{belowskip=0pt}
\lstset{aboveskip=2pt}
\lstset{basicstyle=\scriptsize}

In this section, we illustrate in detail our implementation of \ezsegway in P4. 
We remind the reader that a P4 program consists of a ``control'' part, a set of 
``tables'', and a set of ``actions''. The control part is responsible for executing 
the logic of the packet processing opertations. A table consists of a set of 
entries that perform an action if the processed packet header matches some 
conditions. Finally, an action consists of a set of read\slash write operations 
performed on the packet and on the internal state of the switch.

Figure~\ref{fig.p4-header-type} shows the definition in P4 of all the header types used in the messages exchanged in \ezsegway, \ie, $\MsgInstallUpdate, \MsgGoodToMove$, and $\MsgRemoving$. Namely, since P4 allows the programmer to define a single header type, we define \texttt{inst\_t} so as to contain the union of all the needed fields.  Observe that we define a \texttt{msg\_type} field so that a switch can distinguish among the different types of messages.
For convenience, we also define two additional \texttt{header\_type}, called \texttt{cap\_t} and \texttt{flow\_id\_t}, that are simply used as structured variable types and not as packet headers. These two variables are instantiated with the \texttt{metadata} instruction as the \texttt{cap\_meta} and \texttt{id\_meta}, respectively. We remind the reader that \texttt{metadata} state is destroyed as soon as the processed packet leaves the switch. To maintain state within a switch, we use registers, which are arrays of
bits. We use multiple registers -- one for each state variable -- because P4 does not allow to create a register of composite data types.
 Each \texttt{register} instruction is instantiated FLOW\_COUNT or LINK\_COUNT times, where those constants denote the maximum number of flows and links at a switch, respectively.  
We remind the reader that only registers can be accessed from the ``action'' part of the P4 program, which we describe later in this section. 

\begin{figure}[tp]
\begin{lstlisting}[frame=single,breaklines=true,breakatwhitespace=true,prebreak=\mbox{\space\tiny$\hookleftarrow$},framexleftmargin=4pt,language=C,morekeywords={fields, if, else, header_type, header, control, apply, metadata},escapechar=@]

header_type inst_t{
  fields{
    msg_type: 2; /* 1: @$\MsgInstallUpdate$@,
                    2: @$\MsgGoodToMove$@,
                    3: @$\MsgRemoving$@ */
    id: 16;
    vol: 4;
    old_lnk: 8;
    new_lnk: 8;
    pre_lnk: 8;
    is_end_sw: 1;
    version: 8;
  }
}

header inst_t inst;
metadata inst_t inst_meta;

header_type cap_t{
  fields{
    old_cap: 16;
    new_cap: 16;
  }
}
metadata cap_t cap_meta;

header_type flow_id_t{
  fields{
    id: 8;
    gtm_recv: 1;
    done: 1;
  }
}
metadata flow_id_t id_meta;

register vols_register{
  width: 16; instance_count : FLOW_COUNT;
}
register out_lnks_register{
  width: 16; instance_count : FLOW_COUNT;
}
register pre_lnks_register{
  width: 16; instance_count : FLOW_COUNT;
}
register old_lnks_register{
  width: 16; instance_count : FLOW_COUNT;
}
register new_lnks_register{
  width: 16; instance_count : FLOW_COUNT;
}
register old_pre_register{
  width: 16; instance_count : FLOW_COUNT;
}
register is_end_sws_register{
  width: 1; instance_count : FLOW_COUNT;
}
register is_init_sws_register{
  width: 1; instance_count : FLOW_COUNT;
}
register versions_register{
  width: 8; instance_count : FLOW_COUNT;
}
register gtm_recvs_register{
  width: 1; instance_count : FLOW_COUNT;
}
register dones_register{
  width: 1; instance_count : FLOW_COUNT;
}
register caps_register{
  width: 16; instance_count : LINK_COUNT;
}
\end{lstlisting}
\caption{Header type, metadata and registers in P4.}
\label{fig.p4-header-type}
\end{figure}

In our P4 program, \ezsegway packets are processed by a set of {\em control} functions, which form the so-called {\em control flow}. In Figure~\ref{fig.p4-controls}, we describe all the control functions needed in \ezsegway. When a packet enters a switch, the \texttt{ingress} control function is the first being called. We use this control function to distinguish among the different message types. The three controls $ctrl\_inst$, $ctrl\_gtm$, and $ctrl\_rm$ are used to process $\MsgInstallUpdate$, $\MsgGoodToMove$ and $\MsgRemoving$, respectively. The control function $ctrl\_move\_flow$ ($ctrl\_rm\_flow$) is used in $ctrl\_gtm$ ($ctrl\_rm$) to move a flow to its new path (to remove the flow state from its old path).
We remind the reader that logic statements can only be used in the control flow program and they are not allowed to be used in the rest of the program. The \texttt{apply} instruction is used to redirect the execution to the table part of the P4 program.
\begin{figure}[tp]
\begin{lstlisting}[frame=single,breaklines=true,breakatwhitespace=true,prebreak=\mbox{\space\tiny$\hookleftarrow$},framexleftmargin=4pt,language=C,morekeywords={fields, if, else, header_type, header, control, apply, action},escapechar=@]

control ingress{
  if (valid(inst)){
    if (inst@.@msg_type==1){
      ctrl_inst();
    }
    else if (inst@.@msg_type==2){
      ctrl_gtm();
    }
    else if (inst@.@msg_type==3){
      ctrl_rem();
    }
  }
}

control ctrl_inst{
  apply(save_inst);
  if (inst@.@is_end_sw == 1){
    apply(read_inst);
    apply(send_gtm_inst);
  }
}

control ctrl_rm{
  if (inst_meta@.@version==inst@.@version){
    ctrl_rm_flow();
  }
}

control ctrl_gtm{
  apply(tbl_read_inst_gtm);
  if (inst_meta@.@version==inst@.@version){
    apply(tbl_wrt_gtm_recv);
    apply(tbl_read_cap);
    if (cap_meta@.@new_cap > inst@.@vol){
      ctrl_move_flow();
      apply(tbl_read_inst_id_0);
      if (id_meta@.@gtm_recv==1 and
          cap_meta@.@new_cap > inst_meta@.@vol){
        ctrl_move_flow_id_0();
      }
      [...]
      apply(tbl_read_inst_id_N);
      if (id_meta@.@gtm_recv==1 and
          cap_meta@.@new_cap > inst_meta@.@vol){
        ctrl_move_flow_id_N();
      }
    }
  }
}

control ctrl_move_flow{
  apply(tbl_move_flow);
  if (inst@.@pre_lnk@!=@-1){
    apply(tbl_send_gtm);
  }
  if (inst@.@old_link@!=@-1){
    apply(tbl_rm_flow_fr_add);
    apply(tbl_send_rm_fr_add);
  }
}

control ctrl_rm_flow{
  apply(tbl_rm_flow);
  if (inst@.@old_lnk@!=@-1){
    apply(tbl_send_rm);
  }
}
\end{lstlisting}
\caption{Controls in P4.}
\label{fig.p4-controls}
\end{figure}

We now show in Figure~\ref{fig.p4-tables}) the different table functions that are called in the control part of the P4 program. In our program, each \texttt{table} statement simply consists of a series of actions that are applied to the processed packet. According to $P4$'s convention, every table must be called only once in the whole program.

\begin{figure}[tp]
\begin{lstlisting}[frame=single,breaklines=true,breakatwhitespace=true,prebreak=\mbox{\space\tiny$\hookleftarrow$},framexleftmargin=4pt,language=C,morekeywords={fields, if, else, header_type, header, control, table, actions},escapechar=@]

table tbl_wrt_gtm_recv{
  actions{ write_gtm_recv; } 
}

table tbl_save_inst{
  actions{ save_inst; }
}

table tbl_send_gtm{
  actions{ send_gtm; }
}

table tbl_send_rm{
  actions{ send_rm; }
}

table tbl_send_rm_fr_add{
  actions{ send_rm; }
}

table tbl_send_gtm_inst{
  actions{ update_id_meta; send_gtm; }
}

table tbl_rem{
  actions{ send_rem; }
}

table tbl_read_inst{
  actions{ read_inst; }
}

table tbl_read_inst_gtm{
  actions{ read_inst; }
}

table tbl_read_cap{
  actions{ read_cap; }
}

table tbl_move_flow{
  actions{ add_flow; }
}

table tbl_rm_flow{
  actions{ remove_flow; }
}

table tbl_rm_flow_fr_add{
  actions{ remove_flow; }
}

table tbl_save_gtm_recv{
  actions{write_gtm_recv; }
}
\end{lstlisting}
\caption{All tables in P4.}
\label{fig.p4-tables}
\end{figure}

\begin{figure}[t!]
\begin{lstlisting}[frame=single,breaklines=true,breakatwhitespace=true,
 prebreak=\mbox{\space\tiny$\hookleftarrow$},framexleftmargin=4pt,language=C,
 morekeywords={fields,
	 if, else, header_type, header, control, apply, action},escapechar=@]
action save_inst(){
 vols_reg@[@inst@.@id@]=@inst@.@vol; 
 old_lnks_reg@[@inst@.@id@]=@inst@.@old_lnk;
 new_lnks_reg@[@inst@.@id@]=@inst@.@new_lnk;
 is_end_sws_reg@[@inst@.@id@]=@inst@.@is_end_sw;
 pre_lnks_reg@[@inst@.@id@]=@inst@.@pre_lnk;
 versions_reg@[@inst@.@id@]=@inst@.@version;
 gtm_recvs_reg@[@inst@.@id@]=@0;
 dones_reg@[@inst@.@id@]=@0;
}
	
action forward(port){
 standard_metadata@.@egress_spec=port@;@
}
	
action send_gtm(){
 inst@.@msg_type@=@2;
 standard_metadata@.@egress_spec=
     inst_meta@.@pre_lnk@;@
}
	
action send_rm(){
 inst@.@msg_type@=@3;
 standard_metadata@.@egress_spec=
     inst_meta@.@pre_lnk@;@
}
	
action update_id_meta(){
 id_meta@.@id@=@inst@.@id;
}
	
action read_inst(){
 inst_meta@.@vol@=@vols_reg@[@id_meta@.@id@];@   
 inst_meta@.@old_lnk@=@old_lnks_reg@[@id_meta@.@id@];@
 inst_meta@.@new_lnk@=@new_lnks_reg@[@id_meta@.@id@];@
 inst_meta@.@is_end_sw@=@is_end_sws_reg@[@id_meta@.@id@];@
 inst_meta@.@pre_lnk@=@pre_lnks_reg@[@id_meta@.@id@];@
 inst_meta@.@version@=@versions_reg@[@id_meta@.@id@];@
 id_meta@.@gtm_recv@=@gtm_recvs_reg@[@id_meta@.@id@];@
 id_meta@.@done@=@dones_reg@[@id_meta@.@id@];@
}
	
action write_gtm_recv(){
 gtm_recvs_reg@[@id_meta@.@id@]=@1;
}
	
action read_cap(){
 cap_meta@.@old_cap@=@caps_reg@[@inst@.@old_lnk@]@;
 cap_meta@.@new_cap@=@caps_reg@[@inst@.@new_lnk@]@;
}
	
action add_flow(){
 out_lnks_reg@[@id_meta@.@id@]@=inst_meta@.@new_lnk;
 caps_reg@[@inst_meta@.@new_lnk@]@=cap_meta@.@new_cap
 - inst_meta@.@vol@;@
 dones_reg@[@id_meta@.@id@]=@1;
}
	
action remove_flow(){
 out_lnks_reg@[@id_meta@.@id@]=@-1;
 caps_reg@[@inst_meta@.@old_lnk@]=@cap_meta@.@old_cap
 @+@inst_meta@.@vol@;@
 dones_reg@[@id_meta@.@id@]=@1;
}
\end{lstlisting}
\caption{Actions in P4.}
\label{fig.p4-actions}
\end{figure}
Finally, Figure~\ref{fig.p4-actions} shows all the action functions that are used to process a packet. In P4, each action consists of a sequence of read/write operations from/to the registers. For instance, the \texttt{add\_flow} action performs the following three actions: i)  it updates the new outgoing link of the processed flow, ii) it registers that a flow update operation has been performed, and iii) it updates the residual bandwidth capacity of the new outgoing link.

$P4$ does not allow the flow execution to call the same \texttt{table} instruction multiple times. As such, we generated our code in such a way that the preprocessor will generate a table for each specific flow id, as shown in Figure~\ref{fig.p4-preprocessing}.
\begin{figure}[t!]
\begin{lstlisting}[frame=single,breaklines=true,breakatwhitespace=true,prebreak=\mbox{\space\tiny$\hookleftarrow$},framexleftmargin=4pt,language=C,morekeywords={fields, if, else, header_type, header, control, apply, action, define, table, actions},escapechar=@]
#define tbl_send_gtm_id(x)\
table tbl_send_gtm_id_ ## x{\
  actions{ send_gtm; }\
}

#define tbl_read_inst_id(x)\
table tbl_read_inst_id_ ## x}\
  actions{ update_id_meta; read_inst; }\
}

#define table tbl_move_flow_id(x)\
table tbl_move_flow_id_ ## x{\
  actions{ move_flow; }\
}

#define ctrl_move_flow_id(x)\
control ctrl_move_flow_id_ ## x{\
  apply(tbl_move_flow_id_ ## x);\
  if (inst@.@pre_lnk == -1){\
    apply(tbl_send_gtm_id_ ## x);\
  }\
}
\end{lstlisting}
\caption{Preprocessing code.}
\label{fig.p4-preprocessing}
\end{figure}
}
\fi

\end{document}